\documentclass{lmcs} 
\pdfoutput=1

\usepackage{lastpage}
\lmcsdoi{19}{4}{3}
\lmcsheading{}{\pageref{LastPage}}{}{}%
{Nov.~28,~2022}{Oct.~18,~2023}{}


\usepackage[utf8]{inputenc}
\usepackage{amssymb, colonequals, float, booktabs}
\usepackage{hyperref}

\def\ie{{i.e.}}
\def\eg{{\em e.g.}}

\newcommand{\G}{\Gamma}
\newcommand{\D}{\Delta}
\newcommand{\imp}{\rightarrow}
\newcommand{\ndef}{ \not \bowtie}
\newcommand{\isdef}{ \bowtie }
\newcommand{\impdef}{ \mathrel{\ltimes} }
\newcommand{\eqdef}{ \bowtie } 
\newcommand{\trule}[3]{\frac{#1 \hspace{1cm} #2}{#3}}
\newcommand{\orule}[2]{\frac{#1 }{ #2}}
\newcommand{\Sin}{\in^*}
\newcommand{\notSin}{\notin^*}
\newcommand{\canmod}{\mathcal{C}^c}
\newcommand{\C}{\mathcal{C}}
\newcommand{\LMC}{\mathbb{G}} 
\newcommand{\ce}{\colonequals}
\newcommand{\cce}{\coloncolonequals}
\newcommand{\M}{\widehat{K}}
\renewcommand{\phi}{\varphi}

\newcommand{\VV}{\mathcal{V}}
\newcommand{\FF}{\mathcal{F}}
\newcommand{\Sfive}{\mathsf{S5}}
\newcommand{\nisdef}{\not\isdef}

\begin{document}

\title{Impure Simplicial Complexes: Complete Axiomatization}
\thanks{R.~Kuznets and R.~Randrianomentsoa are funded by the Austrian Science Fund (FWF)  ByzDEL project (P33600).}

\author[R.~Randrianomentsoa]{Rojo {Randrianomentsoa}\lmcsorcid{0000-0002-4553-5450}}[a]
\author[H.~van~Ditmarsch]{Hans {van~Ditmarsch}\lmcsorcid{0000-0003-4526-8687}}[b]
\author[R.~Kuznets]{Roman {Kuznets}\lmcsorcid{0000-0001-5894-8724}}[a]

\address{TU Wien, Austria}
\email{rojo.randrianomentsoa@tuwien.ac.at, roman.kuznets@tuwien.ac.at}

\address{University of Toulouse, CNRS, IRIT, France}
\email{hans.van-ditmarsch@irit.fr}
	
\begin{abstract}
\noindent Combinatorial topology is used in distributed computing to model concurrency and asynchrony. The basic structure in combinatorial topology is the simplicial complex, a collection of subsets called simplices of a set of vertices, closed under containment. Pure simplicial complexes describe message passing in asynchronous systems where all processes~(agents) are alive, whereas impure simplicial complexes describe message passing in synchronous systems where processes may be dead~(have~crashed). Properties of impure simplicial complexes can be described in a three-valued multi-agent epistemic logic where the third value represents  formulae  that are undefined, \eg,~the knowledge and local propositions of dead agents. In this work we present an axiomatization for the logic of the class of impure complexes and show soundness and completeness. The completeness proof involves the novel construction of the canonical simplicial model and requires a careful manipulation of undefined formulae.
\end{abstract}

\maketitle

\section{Introduction}
\label{intro}

In this contribution, we present an axiomatization called~$\Sfive^{\isdef}$ for an epistemic logic and show that it is sound and complete for the class of structures known as \emph{impure simplicial complexes}.  
Not surprisingly given its name, $\Sfive^{\isdef}$~is a variant of the well-known multi-agent epistemic logic~$\Sfive$~\cite{hintikka:1962}. 
And in view of that, it may not come as a complete surprise that impure simplicial complexes, a notion from distributed computing~\cite{herlihyetal:2013}, correspond to certain multi-agent Kripke models~\cite{GoubaultLR22,vDitKuz22arXiv}, the usual vehicle to interpret knowledge on. Impure complexes model that agents (or~processes) may have crashed, in other words, may not be alive but dead. 
A live agent may be uncertain whether other agents are dead or alive but may wish to reason about them regardless. In~\cite{vDitKuz22arXiv} (extending~\cite{vDit21WoLLIC}), it was argued that  for a dead agent it should be undefined what it knows or does not know. 
This requires a three-valued logical semantics for the knowledge modality~$K_a$ (for~`agent~$a$ knows~that') and, consequently, a variation of the standard axiomatization~$\Sfive$ wherein, for example, 
the normality axiom~\textbf{K}, \ie~\mbox{$K_a (\phi \imp \psi) \imp (K_a \phi \imp K_a \psi)$}, and modus ponens~\textbf{MP}, \ie~``from~$\phi$~and~$\phi\imp\psi$, infer~$\psi$,'' do not always hold. 
Still~--- we do not wish to put off the reader too much at this stage~--- the usual three properties of knowledge, \ie~veracity, positive introspection, and negative introspection, are  valid in this three-valued semantics. 
In~\cite{vDitKuz22arXiv}, this semantics was extensively motivated. However, no complete axiomatization was given. 
The underlying work is a follow-up investigation providing that axiomatization, and only that, and nothing else. 
In order to make the work self-contained we first sketch the area and the results already obtained in the literature. 

\paragraph{Knowledge and complexes}
Combinatorial topology~\cite{herlihyetal:2013} has been used in distributed computing to model concurrency and asynchrony since~\cite{FischerLP85,luoietal:1987,BiranMZ90}, with higher-dimensional topological properties entering the picture in~\cite{HS99,SZ00}. 
The basic structure in combinatorial topology is the \emph{simplicial complex}, a~collection of subsets called \emph{simplices} of a set of \emph{vertices}, closed under containment. 
Geometric manipulations such as subdivision have natural combinatorial counterparts. 

Epistemic logic investigates knowledge and belief, and change of knowledge and belief, in multi-agent systems. 
A foundational study is~\cite{hintikka:1962}. 
Knowledge change was extensively modeled in temporal epistemic logics~\cite{Pnueli77,halpernmoses:1990,dixonetal.handbook:2015} and in dynamic epistemic logics~\cite{baltagetal:1998,hvdetal.del:2007,moss.handbook:2015}.

A vertex of a complex represents the local state of an agent or process. 
Given a complex, an agent may be uncertain about the local state of other agents. 
This happens when a vertex for this agent is contained in different simplices that contain vertices for the other agents but for different local states of those agents. 
Similarly, higher-dimensional topological properties involve higher-order knowledge or group knowledge. 
There is, therefore, a natural relation between epistemic reasoning and simplicial complexes. 

An epistemic logic interpreted on \emph{pure} simplicial complexes (\ie~where all processes are alive) was proposed in~\cite{ledent:2019,GoubaultLR21}, which was finalized in~\cite{goubaultetal_postdali:2021}. 
It is axiomatized by the multimodal logic~$\Sfive$ augmented with the so-called \emph{locality axiom} stating that all agents know their local state, 
\ie~$K_a p_a \vee K_a \neg p_a$ for all agents~$a$ and local {propositional} variables~$p_a$ for that agent.

Common knowledge and distributed knowledge were also considered in these works, and even the application of novel group epistemic notions such as common distributed knowledge~\cite{Baltag20} to the setting of simplicial semantics, exactly in view of describing higher-dimensional topological properties (for example, whether a complex is a manifold)~\cite{ledent:2019,hvdetal.simpl:2021}.\looseness=-1

The dynamics of knowledge was another focus of those works and of additional works~\cite{PflegerS18,goubaultetal_postdali:2021,diego:2021}: 
the action models of~\cite{baltagetal:1998} can be used to model distributed computing tasks and algorithms. 
Action models even appeared in combinatorial topological incarnations as simplicial complexes in~\cite{ledent:2019} and in~\cite{hvdetal.simpl:2021}. 
Additionally, \cite{diego:2021}~proposed novel dynamic objects called communication patterns that represent arbitrary message exchange and are based on similar proposals in dynamic epistemics~\cite{AgotnesW17,Baltag20}. 
Communication patterns are different from action models, and their relation is {complicated}.

Pure complexes and their temporal development describe \emph{asynchronous} systems. 
All processes then remain potentially active~(alive), even when they have crashed: in asynchronous systems~\cite{HS99} crashed processes can be modeled as being infinitely slow. 
In \emph{synchronous} systems~\cite{HerlihyRT09} crashed processes can be identified~--- for example, in message passing, with timeouts. 
This makes it possible to use \emph{impure} simplicial complexes, wherein some processes are dead, to represent them. 
For a detailed example, see~\cite[Section~13.5.2]{herlihyetal:2013} {(see also~\cite{HerRajTut} for more extended discussion)}. 
Impure simplicial complexes are considerably smaller in terms of their number of faces.

Epistemic logics interpreted on impure simplicial complexes were proposed in~\cite{vDit21WoLLIC}, of which the extended version, henceforth the only cited one, is~\cite{vDitKuz22arXiv},\footnote{All statements from~\cite{vDitKuz22arXiv} we rely on are also present and proved in~\cite{vDit21WoLLIC}, unless stated otherwise.} and also in~\cite{GoubaultLR22}.

In~\cite{vDitKuz22arXiv}, a three-valued modal logical semantics was proposed for an epistemic logical language interpreted in faces of impure simplicial complexes that are decorated with agents and with local propositional variables. 
The third value stands for `undefined.' 
The following are undefined: 
\begin{itemize}
\item dead processes cannot know or be ignorant of any proposition; 
\item live processes cannot know or be ignorant of factual propositions involving processes they know to be dead. 
\end{itemize}
The issue of the definability of formulae has to be handled with care: standard notions such as validity, equivalence, and the interdefinability of dual modalities and non-primitive propositional connectives have to be properly addressed. 
The epistemic modality is a three-valued version of $\Sfive$~knowledge. 
No complete axiomatization was given in~\cite{vDitKuz22arXiv}. 

In~\cite{GoubaultLR22}, a different, classical two-valued modal logical semantics was proposed, also interpreted in impure simplicial complexes. 
The authors reduce the setting to one where processes may not know their local state, and that can then be treated in standard multi-agent modal logic: 
they  axiomatize the logic as multimodal~$\mathsf{KB4}$. 
Exciting further developments of this framework can be found in~\cite{lics23}: 
namely,  distributed knowledge modalities interpreted on simplicial sets (also known as pseudo-complexes), a generalization of simplicial complexes wherein a face may have multiple occurrences.

The works~\cite{GoubaultLR22,vDitKuz22arXiv} also showed that, whereas pure complexes correspond to Kripke models with equivalence relations, impure complexes correspond to Kripke models with partial equivalence relations (symmetric and transitive relations that need not be reflexive).

The novel epistemic notions proposed for impure complexes are akin to  knowledge but are, therefore, not exactly $\Sfive$~knowledge. 
Consider the {\bf T}~axiom $K_a\phi\imp\phi$. 
It is an axiom in~$\Sfive^{\isdef}$ (and in~\cite{vDitKuz22arXiv}) but its meaning is `known formulae cannot be false' rather  than `known formulae  must be true' (because they may be undefined instead). 
\textbf{T}~is not an axiom in $\mathsf{KB4}$~\cite{GoubaultLR22}, but, conveniently, in that work $K_a \bot$~denotes that `agent~$a$ is dead.' 
Various other synchronous and asynchronous epistemic notions have been proposed in distributed computing, relating in yet other ways to modal logics of knowledge and  belief. 
We recall `belief as defeasible knowledge' of~\cite{MosesS93}, 
knowledge in interpreted systems of~\cite{DworkM90,halpernmoses:1990}, 
distinctions between knowledge and belief presented in~\cite{HalpernSS09a}, and 
encoding byzantine faults through the hope modality, an epistemic notion weaker than belief~\cite{Fruzsa:22}, among others.\looseness=-1

\paragraph{Examples of pure complexes and corresponding Kripke models} Some of our readers may be familiar with simplicial complexes, while other readers may be familiar with Kripke models. As the correspondence between these semantic structures has only recently been proposed~\cite{goubaultetal_postdali:2021}, we provide some simple examples illustrating how.

Figure~\ref{fig} depicts some pure simplicial complexes and corresponding Kripke models:
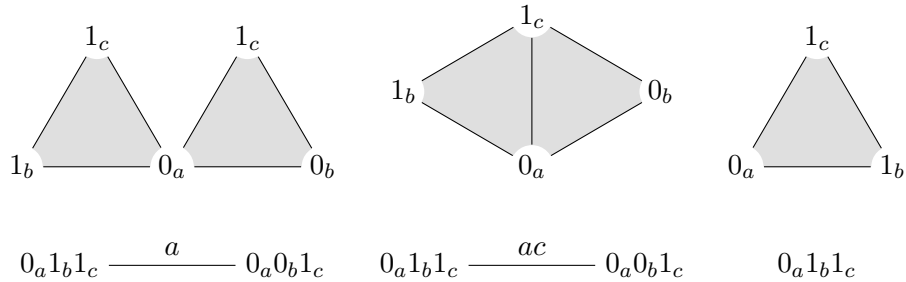
\begin{figure}[ht]
\begin{tabular}{ccc}
\begin{tikzpicture}[round/.style={circle,fill=white,inner sep=1}]
\fill[fill=gray!25!white] (2,0) -- (4,0) -- (3,1.71) -- cycle;
\fill[fill=gray!25!white] (0,0) -- (2,0) -- (1,1.71) -- cycle;

\node[round] (b1) at (0,0) {$1_b$};
\node[round] (b0) at (4,0) {$0_b$};
\node[round] (c1) at (3,1.71) {$1_c$};
\node[round] (lc1) at (1,1.71) {$1_c$};
\node[round] (a0) at (2,0) {$0_a$};

\draw[-] (b1) -- (a0);
\draw[-] (b1) -- (lc1);
\draw[-] (a0) -- (lc1);
\draw[-] (a0) -- (b0);
\draw[-] (b0) -- (c1);
\draw[-] (a0) -- (c1);
\end{tikzpicture}
&
\begin{tikzpicture}[round/.style={circle,fill=white,inner sep=1}]
\fill[fill=gray!25!white] (2,0) -- (2,2) -- (3.71,1) -- cycle;
\fill[fill=gray!25!white] (0.29,1) -- (2,0) -- (2,2) -- cycle;

\node[round] (b1) at (.29,1) {$1_b$};
\node[round] (b0) at (3.71,1) {$0_b$};
\node[round] (c1) at (2,2) {$1_c$};
\node[round] (a0) at (2,0) {$0_a$};

\draw[-] (b1) -- (a0);
\draw[-] (b1) -- (c1);
\draw[-] (a0) -- (b0);
\draw[-] (b0) -- (c1);
\draw[-] (a0) -- (c1);
\end{tikzpicture}
& 
\begin{tikzpicture}[round/.style={circle,fill=white,inner sep=1}]
\fill[fill=gray!25!white] (2,0) -- (4,0) -- (3,1.71) -- cycle;

\node[round] (b0) at (4,0) {$1_b$};
\node[round] (c1) at (3,1.71) {$1_c$};
\node[round] (a0) at (2,0) {$0_a$};

\draw[-] (a0) -- (b0);
\draw[-] (b0) -- (c1);
\draw[-] (a0) -- (c1);
\end{tikzpicture} 
\\ 
&& \\
\begin{tikzpicture}
\node (010) at (.5,0) {$0_a1_b1_c$};
\node (001) at (3.5,0) {$0_a0_b1_c$};
\draw[-] (010) -- node[above] {$a$} (001);
\end{tikzpicture}
&
\begin{tikzpicture}
\node (010) at (.5,0) {$0_a1_b1_c$};
\node (001) at (3.5,0) {$0_a0_b1_c$};
\draw[-] (010) -- node[above] {$ac$} (001);
\end{tikzpicture}
&
\begin{tikzpicture}
\node (010) at (.5,0) {$0_a1_b1_c$};
\end{tikzpicture}
\end{tabular}
\caption{Pure simplicial complexes and corresponding Kripke models}
\label{fig}
\end{figure}

These simplicial complexes are for three agents~$a$, $b$,~and~$c$. 
The vertices of a simplex are required to be labeled, or colored, with different agents. 
Therefore, a maximum-size simplex, called \emph{facet},  consists of three vertices, described as having \emph{dimension~$2$}. 
These are the triangles in the figure. 
For two agents we get lines/edges, for four agents we get tetrahedra,~etc. 
In the figure, the vertices are not named, to simplify the exposition. 
A label like~$0_a$ on a vertex represents that it is a vertex for agent~$a$ and that agent~$a$'s local state has value~$0$,~etc. 
We can see this as the boolean value of a local proposition, where $0$~means false and $1$~means true. 
In other figures (where the value of a proposition is irrelevant) we may only provide agent labels or vertex names. 
Simplicial complexes where vertices are decorated with agents and with propositions will later be called simplicial \emph{models}.

The Kripke models are also defined for three agents. 
The {worlds} that are indistinguishable for agent~$a$ are connected by an $a$-labeled link; in other words, for each agent we assume reflexivity, transitivity, and symmetry of the usual accessibility relations. 

A triangle corresponds to a {world}  in a three-agent Kripke model. 
Together the local states of the triangle determine the {world} of the Kripke model, and the values of the local state variables together determine the valuation in that {world}. 
For example, in {worlds} labeled~$0_a1_b1_c$ in Figure~\ref{fig} agent~$a$'s value is~$0$, $b$'s~is $1$, and $c$'s~is $1$. 
The single triangle corresponds to the singleton reflexive world below it. 
With two triangles, if they  intersect only in an $a$-colored vertex, \ie~are \emph{$a$-adjacent}, it means that agent~$a$ cannot distinguish these (global)~states (as in the corresponding Kripke model), making $a$~uncertain about the value of~$b$. 
If the triangles intersect in two vertices colored~$a$~and~$c$, both~$a$~and~$c$ are uncertain about the value of~$b$.

The current state of the distributed system is represented by a distinguished facet of the simplicial complex, just as we need a distinguished {world} in a Kripke model in order to evaluate propositions. 
For example, in the leftmost triangle, as well as in the leftmost {world}, the value of~$b$ is~$1$ but $a$~is uncertain whether the value of~$b$ is~$0$~or~$1$. 
In other words, $a$~does not \emph{know} whether the value of~$b$ is~$0$~or~$1$. 
In the same triangle, $b$~knows that its value is~$1$ and all three agents know that the value of~$c$ is~$1$.

\paragraph{Examples of impure  complexes and corresponding Kripke models}
We continue with several examples of impure complexes, also in relation to pure complexes we have already seen above. 
They are depicted in Figure~\ref{fig2}. See~\cite[Section~13.5.2]{herlihyetal:2013} for a standard example of an impure simplicial complex from distributed computing.

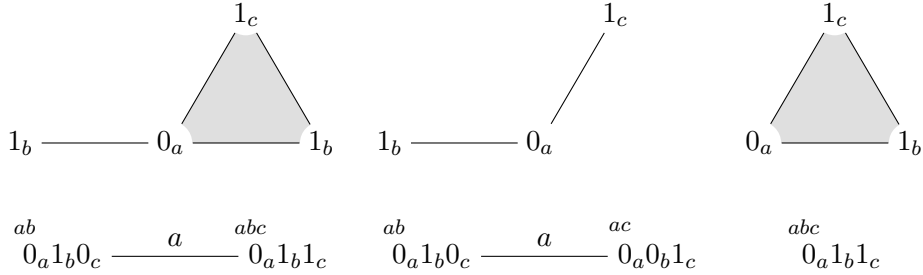
\begin{figure}[ht]
\begin{tabular}{cccccc}
\begin{tikzpicture}[round/.style={circle,fill=white,inner sep=1}]
\fill[fill=gray!25!white] (2,0) -- (4,0) -- (3,1.71) -- cycle;

\node[round] (b1) at (0,0) {$1_b$};
\node[round] (b0) at (4,0) {$1_b$};
\node[round] (c1) at (3,1.71) {$1_c$};
\node[round] (a0) at (2,0) {$0_a$};

\draw[-] (b1) --  (a0);
\draw[-] (a0) -- (b0);
\draw[-] (b0) -- (c1);
\draw[-] (a0) -- (c1);
\end{tikzpicture} 
&
\begin{tikzpicture}[round/.style={circle,fill=white,inner sep=1}]
\node[round] (b1) at (0,0) {$1_b$};
\node[round] (b0) at (4,0) {\color{white}$0_b$};
\node[round] (c1) at (3,1.71) {$1_c$};
\node[round] (a0) at (2,0) {$0_a$};

\draw[-] (b1) --  (a0);
\draw[-] (a0) --  (c1);
\end{tikzpicture} 
&
\begin{tikzpicture}[round/.style={circle,fill=white,inner sep=1}]
\fill[fill=gray!25!white] (2,0) -- (4,0) -- (3,1.71) -- cycle;

\node[round] (b0) at (4,0) {$1_b$};
\node[round] (c1) at (3,1.71) {$1_c$};
\node[round] (a0) at (2,0) {$0_a$};

\draw[-] (a0) -- (b0);
\draw[-] (b0) -- (c1);
\draw[-] (a0) -- (c1);
\end{tikzpicture} 
\\ 
&& \\
\begin{tikzpicture}
\node (010l) at (0,0.4) {\scriptsize$ab$};
\node (001l) at (3,0.4) {\scriptsize$abc$};
\node (010) at (.5,0) {$0_a1_b0_c$};
\node (001) at (3.5,0) {$0_a1_b1_c$};
\draw[-] (010) -- node[above] {$a$} (001);
\end{tikzpicture}
& 
\begin{tikzpicture}
\node (010l) at (0,0.4) {\scriptsize$ab$};
\node (001l) at (3,0.4) {\scriptsize$ac$};
\node (010) at (.5,0) {$0_a1_b0_c$};
\node (001) at (3.5,0) {$0_a0_b1_c$};
\draw[-] (010) -- node[above] {$a$} (001);
\end{tikzpicture} 
&
\begin{tikzpicture}
\node (010) at (0,0) {$0_a1_b1_c$};
\node (001l) at (-.5,.4) {\scriptsize$abc$};
\end{tikzpicture}
\end{tabular}
\caption{Impure simplicial complexes and corresponding Kripke models}
\label{fig2}
\end{figure}

The pure complex on the right encodes that all agents know local states of each other, just as in the same complex on the right of Figure~\ref{fig}. 
It is merely here for comparison.

The impure complex on the left of Figure~\ref{fig2} encodes that agent~$a$ is uncertain whether agent~$c$ is alive. 
When agent~$c$ is dead, in the edge, it does not have a value and propositions like `agent~$c$ knows that $p_b$~is true' are undefined. 
When agent~$c$ is alive, in the triangle, its local variable~$p_c$ is true. 
We, therefore, say that agent~$a$ knows that $p_c$~is true: whenever agent~$c$ is alive and, therefore, local variable~$p_c$ is defined (namely, in the triangle) it is true. 

The impure complex in the middle encodes that agent~$a$ is uncertain whether $b$~or~$c$~has crashed. 
Agent~$a$ no longer considers it possible that all processes are alive. 
In this case, agent~$a$ knows that~$p_b$, justified by the edge~$1_b$---$0_a$ (where~$b$~is alive) and 
agent~$a$ also knows that~$p_c$, justified by the edge~$0_a$---$1_c$ (where~$c$~is alive). 
However, pointedly, agent~$a$ does not know that $p_b \wedge p_c$: 
for that to be true we would need a triangle wherein $a$, $b$,~and~$c$~are all alive. 
This example demonstrates that $K_a \phi \wedge K_a \psi$ does not imply $K_a (\phi\wedge\psi)$.

Below the complexes we find the corresponding Kripke models. 
In this visualization, the upper indices to a {world} make explicit who is alive in this {world}. 
It is merely a depiction of models with symmetric and transitive relations, also known as \emph{partial equivalence relations}. 
Such  a relation is a partial equivalence in the sense that it is an equivalence relation when restricted to the worlds where it is not empty, \ie~in our terms, where the agent is alive. 

In the Kripke model on the left, as $c$~is dead in the $0_a1_b0_c$~world and alive in the $0_a1_b1_c$~world, the relation for agent~$c$ is restricted to the reflexive arrow from~$0_a1_b1_c$ to itself. 
The Kripke model representation, unlike the impure simplicial complex, contains superfluous information: 
even though $c$~is dead in the world~$0_a1_b0_c$, local variable~$p_c$ is false there. Therefore, another representation of the same impure complex is the model~$0_a1_b1_c$~$\overset{a}{\text{---{}---}}$~$0_a1_b1_c$, where again~$c$ is dead in~$0_a1_b1_c$. (Strictly speaking, the number of such representations depends on how many local variables agent~$c$ has.)

Similarly, in the Kripke model in the middle, $c$~is dead on the left and $b$~is dead on the right. 
Accordingly, there are four different Kripke models corresponding to the simplicial complex above it.

These are some examples of pure and impure complexes, and how they are used to model agents' reasoning about each other.

\paragraph{Our results}
We present a complete axiomatization~$\Sfive^{\isdef}$ for a three-valued epistemic semantics interpreted on impure complexes. 
It has some notable differences from the standard multi-agent axiomatization~$\Sfive$, such as non-standard axioms~$\mathbf{K}^{\bowtie}$~and~$\mathbf{K\M}$ instead of~$\mathbf{K}$ and a non-standard derivation rule~$\mathbf{MP}^{\bowtie}$ instead of~$\mathbf{MP}$. 
Some technical features of the completeness proof may be considered novelties of interest. 
First, instead of a canonical Kripke model, we construct a canonical model that is a \emph{canonical simplicial complex}. 
Second, the faces of this canonical simplicial complex are not maximal consistent sets, but what we  call \emph{definability-maximal consistent sets}: they need not be maximal and may be contained in one another, just like the faces of simplicial complexes in general.

\paragraph{Outline}
Section~\ref{ISM_and_Def} recalls the impure simplicial semantics of~\cite{vDitKuz22arXiv} and obtains additional results for formula definability that are crucial in this three-valued setting. Section~\ref{axiom_sys} presents axiomatization~$\Sfive^{\isdef}$. 
Section~\ref{soundness} shows that this axiomatization is sound, and Section~\ref{completeness} that it is complete. 
Section~\ref{reflections} describes some remarkable features of the simplicial canonical model. Section~\ref{conclusion} concludes. 
Some of the proofs are relegated to Appendix~\ref{b}, with Appendix~\ref{a} providing auxiliary statements used exclusively in these proofs.

\section{Impure Simplicial Models and Definability}
\label{ISM_and_Def}

We recall the semantics of simplicial complexes from~\cite{vDitKuz22arXiv}, both pure (for fault-free systems) and impure (for systems where byzantine failures are restricted to crashes). 

We consider a finite set~$A$ of \emph{agents} (or~\emph{processes})~$a,b,\dots$ and a  set $P = \bigsqcup_{a \in A} P_a$ of \emph{propositional} \emph{variables} where sets~$P_a$ are countable and mutually disjoint sets of \emph{local variables for agent~$a$}, denoted~$p_a, q_a, p'_a, q'_a, \dots$\footnote{We sometimes use~$\sqcup$ instead of~$\cup$ to emphasize that this is a union of mutually disjoint sets. Similarly, we use~$\subseteq$ for inclusion and $\subset$~for proper inclusion, but may use~$\subsetneq$ to emphasize the inclusion properness.} 
 As usual in combinatorial topology, 
the number~$|A|$ of agents is taken to be~$n+1$ for some~$n \in \mathbb N$, so as to make the dimension of a simplicial complex (to be defined below), which is one less than the number of agents,  equal to~$n$.\looseness=-1

\begin{defi}[Language]
The \emph{language of epistemic logic}  is defined by the grammar $\phi 
\cce 
p_a \mid \neg\phi \mid (\phi\land\phi) \mid K_a \phi$ where $a \in A$ and $p_a \in P_a$.
\end{defi}
Parentheses will be omitted unless confusion results. 
Connectives~$\imp$, $\leftrightarrow$,~and~$\lor$ are defined by abbreviation as usual, as well as $\M_a\phi := \neg K_a \neg \phi$.
Expression~$K_a \phi$ stands for `agent~$a$ knows (that)~$\phi$,' and $\M_a \phi$~stands for `agent~$a$ considers (it possible that)~$\phi$.'  

\begin{defi}[Simplicial model]
A \emph{simplicial model}~$\C$ is a triple\/~$(C,\chi,\ell)$ consisting of: 
\begin{itemize}
\item 
A \emph{(simplicial) complex} $C\ne \varnothing$ is a collection of  \emph{simplices} that are non-empty finite subsets of a given set~$\VV$  of vertices such that $C$~is downward closed (\ie~$X \in C$ and\/ $\varnothing \ne Y \subseteq X$ imply $Y \in C$). Simplices represent partial global states of a distributed system. 
It is required that every vertex form a simplex by itself, \ie~$\bigl\{\{v\} \mid v \in \VV\bigr\} \subseteq C$.
\item 
Vertices represent local states of agents, with a \emph{chromatic map}~$\chi \colon \VV \to A$ assigning each vertex to one of the agents in such a way that  each agent has at most one vertex per  simplex, 
\ie~$\chi(v) = \chi(u)$ for some $v,u \in X \in C$ implies that $v = u$.
For $X \in C$, we define 
\[
\chi(X) \ce \{\chi(v) \mid v \in X\}
\] 
to be the set  of agents in simplex~$X$.
\item 
A \emph{valuation} $\ell \colon \VV \to 2^P$ assigns to each vertex which local variables of the vertex's owner are true in it, \ie~$\ell(v) \subseteq P_a$ whenever $\chi(v) = a$. 
Variables from  $P_a\setminus\ell(v)$ are false in vertex~$v$, whereas variables from $P \setminus P_a$ do not belong to agent~$a$ and cannot be evaluated in $a$'s~vertex~$v$.
The set of variables that are true in a simplex $X \in C$ is given by 
\[
\ell(X)\ce\bigsqcup_{v \in X} \ell(v).
\]
\end{itemize}
If $Y \subseteq X$ for $X, Y \in C$, we say that $Y$~is a \emph{face} of~$X$. 
Since each simplex is a face of itself, we use `simplex' and `face' interchangeably.  
A face~$X$ is a \emph{facet} if it is a maximal simplex in~$C$, \ie~$Y \in C$ and $Y \supseteq X$ imply $Y = X$. 
Facets represent global states of the distributed system, and their set is denoted~$\FF(C)$. 
The \emph{dimension of simplex}~$X$ is\/~$|X|-1$, \eg,~vertices are of dimension\/~$0$,  edges are of dimension\/~$1$,~etc. 
The \emph{dimension of a simplicial complex~(model)} is the largest dimension of its facets. 
A simplicial complex (model) is \emph{pure} if all facets have dimension~$n$, \ie~contain vertices for all agents. 
Otherwise it is \emph{impure}. 
A~\emph{pointed simplicial model} is a pair\/~$(\C,X)$ where $X \in C$. 
\end{defi} 

\begin{rems}\label{rem:simpmodels}\hfill\vspace*{-\baselineskip}
\begin{enumerate}[(i)]
\item
It follows from the above definition that for any simplicial complex~$C$, its set of vertices $\VV = \bigcup_{X \in C} X$. Thus, it is not necessary to specify~$\VV$ when defining a complex~$C$ or a simplicial model $\C=(C,\chi,\ell)$ based on this complex. 
Further, complex~$C$ can be computed from the set~$\FF(C)$ of its facets by taking the downward closure with respect to non-empty sets. 
Thus, we typically forgo defining~$\VV$ explicitly and often specify~$\FF(C)$ instead of~$C$.
\item
\label{rem:facet}
Due to the chromatic-map condition, $|X|\leq n+1$ for all faces $X$ of a complex and, hence, each face is contained in a facet. 
\item
It has been common in the literature, including in~\cite{vDitKuz22arXiv}, to consider a more expansive view of pure simplicial models requiring all the facets to have merely the same dimension, without additionally requiring this common dimension to be the largest possible for the given set of agents as we do here. We believe that the distinction based on the presence/absence of dead agents is more meaningful than the previously used distinction based on the stability of the number of dead agents. Indeed, it is the presence of dead agents that gives rise to the question of what these dead agents know. And, as we will see, it is the presence of dead agents rather than fluctuations of their numbers that separates the logic of impure simplicial models from the standard epistemic logic. In particular, both counterexamples to modus ponens from~\cite{vDitKuz22arXiv} (not present in~\cite{vDit21WoLLIC}) and in the proof of Lemma~\ref{ex:noMP} involve  simplicial models with all facets of the same dimension but  with at least one dead agent per facet.
\end{enumerate}
\end{rems}

Since simplicial models are an alternative semantics for epistemic logic, it is useful to understand their relationship to Kripke models. Worlds of a Kripke model correspond to facets of a simplicial model. Vertices represent local states, and, hence, have no direct Kripke representation. However, they play a crucial role in the simplicial variant of indistinguishability. Where two Kripke worlds are indistinguishable for agent $a$, their corresponding facets~$X$~and~$Y$ in a simplicial model must have agent~$a$ in the same local state, meaning that $X$~and~$Y$~must share an $a$-colored vertex, \ie~$a \in \chi(X \cap Y)$, and we call them  \emph{$a$-adjacent}. We also apply $a$-adjacency to faces.

The special case of pure simplicial models (of the maximum dimension $n$) corresponds to the standard epistemic logic $\Sfive_{n+1}$ (augmented with the locality axiom to account for variables being local, as mentioned in the introduction) and is a representation dual to epistemic Kripke models. As noted in~\cite{vDitKuz22arXiv}, however, the case of impure models introduces a new situation where formulae may be undefined in a facet. Indeed, if $X$ is a facet with $b \notin \chi(X) $, which represents a global state where agent $b$ has crashed, then  all local propositional variables of agent~$b$ are considered undefined. This imposes a two-tier system on formulae: one must first determine whether a formula is defined, before one can probe whether it is true or false. Note that we   do not restrict these definitions to facets only.

\begin{defi}[Definability]
\label{def:defin}
Given a simplicial model $\C=(C,\chi,\ell)$ and its face $X \in C$,
	the definability relation\/~$\isdef$  is defined by recursion on the formula construction:
		\begin{equation*}
		\begin{array}{lll}
			\C,X  \isdef p_a &\text{ if{f} } & a \in \chi(X); \\
			\C,X  \isdef \lnot \phi &\text{ if{f} }&  \C, X \isdef \phi; \\
			\C,X  \isdef \phi \land \psi &\text{ if{f} } & \C,X \isdef \phi \text{ and } \C,X \isdef \psi; \\
			\C,X  \isdef K_a \phi &\text{ if{f} }&  \C,Y \isdef \phi \text{ for some } Y \in C \text{ with } a \in \chi(X\cap Y) .
		\end{array}
	\end{equation*}
\end{defi}
\begin{defi}[Truth]
\label{def:truth}
	Given a simplicial model $\C=(C,\chi,\ell)$ and its face $X \in C$, the truth relation\/~$\vDash$  is defined by recursion on the formula construction:
	\begin{equation*}
		\begin{array}{lll}
			\C,X  \vDash p_a &\text{ if{f} }&    p_a \in \ell(X); \\
			\C,X  \vDash \lnot \phi &\text{ if{f} }&  \C,X \isdef \lnot \phi \text{ and }   \C, X \nvDash \phi; \\
			\C,X  \vDash \phi \land \psi &\text{ if{f} }&  \C,X \vDash \phi \text{ and } \C,X \vDash \psi; \\
			\C,X  \vDash K_a \phi &\text{ if{f} }&  \C, X \isdef K_a \phi \text{ and }
			\\
			&&C,Y \vDash \phi \text{ for all } Y \in C  \text{ such that } a \in \chi(X\cap Y) \text{ and } \C, Y \isdef \phi.
		\end{array}
	\end{equation*}
\end{defi}

\begin{rem}
These definitions are equivalent to those from~\cite{vDitKuz22arXiv}, but differ from them as follows:
\begin{itemize}
\item
$K_a$ is used as the primary connective instead of $\M_a$ in~\cite{vDitKuz22arXiv}, which is harmless in light of the duality of $K_a$ and $\M_a$.
\item In~\cite{vDitKuz22arXiv}, $\C, X \vDash p_a$ additionally requires that $a \in \chi(X)$, which  is redundant because $\ell(X)\cap P_a = \varnothing$ if $a \notin \chi(X)$. 
\item In~\cite{vDitKuz22arXiv}, $\C, X \vDash \lnot \phi$ requires $\C, X \isdef \phi$ instead of the equivalent $\C, X \isdef \lnot \phi$ here.
\end{itemize}
These modifications are primarily to avoid the necessity  of changing  the axiomatization of~$\Sfive$ in order to  ensure the Replacement Property in the language based on a $\Diamond$-like modality.
\end{rem}

One of the design choices made is that undefined formulae do not have truth values.\footnote{The next few results are from~\cite{vDitKuz22arXiv}, but are reproved here for this language.}
\begin{lem}
\label{lem:threevalues}
For a simplicial model $\C=(C,\chi,\ell)$, its face $X \in C$, and formula $\phi$:
\begin{enumerate}[(a)]
\item
\label{clause:threevalues_one}	
$\C, X \vDash \phi$
\qquad$\Longrightarrow$\qquad
$\C, X \isdef \phi$.
\item
\label{clause:threevalues_two}
$\C, X \vDash \lnot \phi$
\qquad$\Longrightarrow$\qquad
 $\C, X \isdef \phi$.
\item 
It cannot be that both\quad $\C, X \vDash \phi$\quad and\quad $\C, X \vDash \lnot \phi$.
\end{enumerate}
\end{lem}
\begin{proof}
\hfill
\begin{enumerate}[(a)]
\item
The proof is by induction on the construction of $\phi$. For $\phi = \lnot \psi$ and $\phi = K_a \psi$, the statement follows directly from Definition~\ref{def:truth}.  If  $\C,X \vDash p_a$, then $p_a \in \ell(X)$, hence, $a \in \chi(X)$, meaning that $\C, X \isdef p_a$. If  $\C, X \vDash \psi \land \theta$, then $\C, X \vDash \psi$ and $\C, X \vDash  \theta$. By the induction hypothesis,  $\C, X \isdef \psi$ and $\C, X \isdef \theta$, hence, $\C, X \isdef \psi \land \theta$.
\item
If $\C, X \vDash \lnot \phi$, then $\C, X \isdef  \lnot \phi$ by clause~\ref{clause:threevalues_one}. Hence, $\C, X \isdef \phi$ by  Definition~\ref{def:defin}.
\item
This follows directly from Definition~\ref{def:truth}. \qedhere
\end{enumerate} 
\end{proof}

Thus, for $\C=(C,\chi,\ell)$, $X\in C$, and  $\phi$, there can be three mutually exclusive outcomes:
\begin{enumerate}[(i)]
\item if $\C, X \vDash \phi$, we say that  \emph{$\phi$ is true in $X$};
\item if $\C, X \nisdef \phi$, we say that \emph{$\phi$ is undefined in $X$}, due to  crashes;
\item if $\C, X \isdef \phi$ and $\C, X \nvDash \phi$, \ie~$\C, X \vDash \lnot \phi$, we say  that \emph{$\phi$ is false in $X$}.
\end{enumerate}

The following corollary is easy to obtain from Definitions~\ref{def:defin} and~\ref{def:truth} and  the abbreviation $\M_a \phi \ce \lnot K_a \lnot \phi$:
\begin{cor}
\label{cor:simple_def}
For a simplicial model $\C=(C,\chi,\ell)$, its face $X \in C$, formula $\phi$, and agent~$a$:
\begin{enumerate}[(a)]
\item
$\C, X \isdef \M_a \phi$ if{f}
$\C,Y \isdef \phi$  for some $Y \in C$ such that  $a \in \chi(X\cap Y)$,
\ie~the definability conditions for $\M_a \phi$ and $K_a \phi$ are the same.
\item
\label{eq:trueKhat}
$\C, X \vDash \M_a \phi$ if{f}
$\C,Y \vDash \phi$  for some $Y \in C$ such that  $a \in \chi(X\cap Y)$.
\item
\label{eq:trueK}
$\C, X \vDash K_a \phi$ if{f}
\begin{itemize}
\item $\C,Y \vDash \phi$  for some $Y \in C$ such that  $a \in \chi(X\cap Y)$
 and
 \item
$\C, Z \vDash \phi$  whenever $\C, Z \isdef \phi$ for any $Z \in C$ such that  $a \in \chi(X\cap Z)$.
\end{itemize}
\item
\label{eq:falseK}
$\C, X \vDash \lnot K_a \phi $ if{f}
$\C,Y \vDash \lnot \phi$  for some $Y \in C$ such that  $a \in \chi(X\cap Y)$. \qed
\end{enumerate} 
\end{cor}

Definability and truth for other boolean connectives are even easier to derive. In particular,   $\phi \lor \psi$, $\phi \to \psi$, and $\phi \leftrightarrow \psi$ are defined if{f} $\phi \land \psi$ is, \ie~when both $\phi$ and $\psi$ are defined.\looseness=-1

The proof of the following monotonicity properties from~\cite{vDitKuz22arXiv} applies as is:
\begin{lem}[Monotonicity \cite{vDitKuz22arXiv}]
	\label{lem:defMonotonicity}
	For a simplicial model $\C=(C,\chi,\ell)$, its faces \mbox{$X,Y \in C$} such that $X \subseteq Y$, and formula $\phi$:
	\begin{enumerate}[(a)]
	\item\label{monot_defin}
	$\C,X \isdef \phi$
	\qquad$\Longrightarrow$\qquad
	$\C,Y \isdef \phi$.
	\item\label{monot:true}
	$\C,X \vDash \phi$
	\qquad$\Longrightarrow$\qquad
	$\C,Y \vDash \phi$.
	\item\label{monot_back}
	$\C,Y \vDash \phi$\quad and\quad  $\C, X \isdef \phi$
	\qquad$\Longrightarrow$\qquad
	$\C,X \vDash \phi$. \qed
	\end{enumerate} 
\end{lem}

Since the standard \emph{modus ponens} rule is not valid for impure simplicial models~\cite{vDitKuz22arXiv} (this result is not present in~\cite{vDit21WoLLIC}), the rule has to be weakened based on the notion of \emph{definability consequence} that is similar to logical consequence but  with definability in place of truth. Before we state the axiom system for the logic of impure simplicial models, we formulate this notion and list its relevant properties.

\begin{defi}[Definability consequence]
A formula $\psi$ is a \emph{definability consequence} of a set\/~$\G$ of formulae, notation\/
 $\G \impdef \psi$, if{f} there exists a finite subset\/ $\G' \subseteq \G$  such that for any simplicial model $\C=(C,\chi,\ell)$ and its face $X \in C$, 
\[
(\forall \phi \in \G')\, \C,X \isdef \phi 
\qquad  \Longrightarrow \qquad
\C,X \isdef \psi.
\] 
To simplify  notation we omit\/~$\{$~and\/~$\}$ for finite sets\/~$\Gamma$, where we assume without loss of generality that\/ $\G' = \G$. Note, in particular, that\/ $\varnothing \not\impdef \psi$  for any formula $\psi$. For a set\/~$\D$ of formulae,\/  $\G \impdef \D$ means that\/  $\G \impdef \psi$ for all $\psi \in \D$. 
Formulae $\phi$ and $\psi$ are called \emph{equidefinable}, written $\phi \eqdef \psi$, when both $\phi \impdef \psi$ and $\psi \impdef \phi$.
\end{defi}

\begin{defi}
For a set\/~$\Gamma$ of formulae and agent~$a$, we define
$K_a \G \ce \{K_a \theta \mid K_a \theta \in \G\}$.
\end{defi}

\begin{restatable}[Properties of definability consequence]{lem}{defcon}
	\label{lem:eq_definabilities}
	For an  agent $a$, propositional variable~$p_a$, formulae $\phi, \psi, \theta, \psi_1,\dots, \psi_m$, and formula sets\/~$\Gamma$ and\/~$\Delta$:
\begin{enumerate}[(a)]
\item
\label{defin_monot}
$\G \impdef \phi$
\qquad$\Longrightarrow$\qquad
$\G \cup \D \impdef \phi$.
\item
\label{lem:eq_def_trans}
$\G \impdef \D$\quad and\quad $\D \impdef \phi$
\qquad$\Longrightarrow$\qquad
$\G \impdef \phi$.
\item
\label{from_K_to_atom}
$K_a \phi  \impdef  p_a$.  
\item 
\label{lem:eq_atom_K}
$p_a \eqdef K_a p_a$.
\item 
\label{def:trans}
$K_a \phi \eqdef K_a K_a \phi$.
\item  
\label{def:eucl}
$K_a \phi \eqdef \lnot K_a \phi \eqdef K_a  \lnot K_a \phi$.
\item
\label{lem:if_phi_then_ka_phi}
$K_a \theta, \phi \impdef K_a \phi$.
\item
\label{def_raise_ka}
$ K_a \G, \phi \impdef K_a\psi$
\qquad$\Longrightarrow$\qquad
$K_a\G, K_a\phi \impdef K_a\psi$.
\item
\label{def_raise_ka_neg}
$ K_a \G, \lnot\phi \impdef K_a\psi$
\qquad$\Longrightarrow$\qquad
$K_a\G, \lnot K_a\phi \impdef K_a\psi$.
\item
\label{lem:def_tough_hatK}
$K_a (K_a \psi_1 \land \cdots \land K_a \psi_m) \imp K_a \theta \quad \impdef \quad K_a \bigl(K_a \psi_1 \land \cdots \land K_a \psi_m \imp  \theta\bigr)$.
\item
\label{lem:def_tough_hatKhat}
$K_a (K_a \psi_1 \land \cdots \land K_a \psi_m) \imp \M_a \theta \quad \impdef \quad K_a \bigl(K_a \psi_1 \land \cdots \land K_a \psi_m \imp  \theta\bigr)$.
\item 
\label{lem:def_tough_K}
$ K_a \psi_1 \land \cdots \land K_a \psi_m \quad \eqdef \quad K_a (K_a \psi_1 \land \cdots \land K_a \psi_m)$.
\end{enumerate} 
\end{restatable}
\begin{proof}
The proofs are straightforward by using Definition~\ref{def:defin} and can be found in Appendix~\ref{b}. 
\end{proof}

A closer observation of impure simplicial models uncovers several peculiarities that affect not only the logic but also the notion of validity and even the object language used. 

For instance, if boolean constants $\top$ and $\bot$ are not chosen as primary connectives, it is typically convenient and harmless to add them to the language as abbreviations by choosing a~(fixed but arbitrary) propositional variable $p$ and using $p \lor \neg p$ for $\top$ and/or $p \land \neg p$ for $\bot$. This works if $p$ is a global variable. In our agent-oriented language, however, all propositional variables are local to some agent. Thus, $p_a \lor \neg p_a$ is not going to be true in all face(t)s, as one would expect from  boolean constant~$\top$. Instead, $p_a \lor \neg p_a$ is true whenever agent~$a$~is alive and undefined otherwise. We can use such a formula at most as a local boolean constant~$\top_a$, which cannot be false but need not be true.

In fact, no formula is true in all face(t)s. For any formula $\phi$, it is not hard to construct a single-vertex simplicial model where $\phi$ is undefined. In other words, were we to define validity as truth in all face(t)s of all models, the set of valid formulae would have been empty. On the other hand, it is reasonable to consider  $\top_a$ mentioned above valid in the sense that it is never false. It is this  notion of validity that we  axiomatize in this paper:
\begin{defi}[Validity \cite{vDitKuz22arXiv}]
\label{def:valid}
A formula $\phi$ is valid in the class of impure simplicial models, written\/ $\vDash \phi$,  if{f} for any face $X$ of any simplicial model~$\C$,
\[
\C, X \isdef \phi 
\qquad\Longrightarrow\qquad
\C, X \vDash \phi .
\]
\end{defi}

One might question why we define validity with respect to all faces even though it is facets that  represent global states. Fortunately, this point is moot  as the logic remains the same if the validity is restricted to facets.
\begin{lem}[Validity on facets]
\label{lem:val_facet}
A formula $\phi$ is true in all facets of all impure simplicial models where it is defined if{f}\/~$\vDash \phi$.
\end{lem}
\begin{proof}
The ``if'' direction trivially follows from Definition~\ref{def:valid} as valid formulae are true in all faces where they are defined, including all facets. For the ``only if'' direction, assume that~$\nvDash \phi$, \ie~there exists a face $X$ of a simplicial model $\C$ such that $\C, X \isdef \phi$ but $\C, X \nvDash \phi$, \ie~\mbox{$\C, X \vDash \neg \phi$}. By Lemma~\ref{lem:defMonotonicity}\ref{monot:true}, $\C, Y \vDash \neg \phi$ for any facet $Y \supseteq X$. By Remark~\ref{rem:simpmodels}\ref{rem:facet}, such a facet exists. In this facet,  $\C, Y \isdef \phi$ but $\C, Y \nvDash \phi$, meaning that $\phi$ is not valid with respect to facets either.
\end{proof}

As a matter of a sanity check, it is easy to see that all propositional tautologies are valid in the class of impure simplicial models. Indeed, whenever such a tautology is defined, boolean connectives behave the same way as in classical propositional logic.

\paragraph{Comparing three- and two-valued semantics for impure complexes}
Let us explain the difference between our approach and that of~\cite{GoubaultLR22} by the leftmost model in Figure~\ref{fig2}, which we  call $\C_3$. It consists of two facets, the 1-dimensional edge, let us call it~$X$, to the left and the two-dimensional gray triangle, let us call it $Y$, to the right, which are $a$-adjacent, signifying the fact that agent $a$ is not sure whether agent $c$ is alive (even though this cannot be expressed in the object language). 

The first observation to make is that our model $\C_3$ corresponds to multiple models in the semantics of~\cite{GoubaultLR22}. Indeed, where we treat any local propositional variable of agent~$c$ as undefined in edge $X$ because $c$ is dead there, the two-valued semantics requires to assign one of two truth values to each of $c$'s variables. This means that for a local state consisting of $n$~local variables, one model $\C_3$ corresponds to $2^n$ two-valued models that differ in the local state of the dead agent (infinitely many local variables would give rise to infinitely many distinct two-valued models). 

The differences among these two-valued models go beyond formulae undefined in~$\C_3$, however. The objective changes to the truth values of dead-agent variables result in subjective changes to the knowledge of live agents. In $\C_3$, agent $a$ knows that $p_c$ is true, \ie~$\C_3 \vDash K_a p_c$, because $p_c$ is true whenever defined. In the two-valued semantics of~\cite{GoubaultLR22}, by contrast, $a$~may or may not know $p_c$ to be true depending on whether $p_c$ is true at $X$, \ie~$\C'_2 \vDash K_a p_c$ while $\C''_2 \vDash \neg K_a p_c$ depending on how $p_c$ is evaluated at $X$, where $c$ is dead, in two-valued models~$\C'_2$~vs.~$\C''_2$.  Thus, a formula valid in three values may be validated or invalidated by the choice of truth values of dead-agent formulae in two values.

It can also happen that a formula valid in three values must be invalid in two values. Consider, for instance, a model $\C^{X-b}_3$ that differs from $\C_3$ only in the value of $p_b$ in edge $X$, \ie~$\C_3, X \vDash p_b$  but $\C^{X-b}_3, X \nvDash p_b$. Then we have $\C^{X-b}_3 \vDash K_a(p_b \land p_c)$ because $p_b \land p_c$ is true in the only facet it is defined, at $Y$. By contrast, $\C^{X-b}_2, X \nvDash p_b \land p_c$ no matter how $p_c$~is evaluated, resulting in $\C^{X-b}_2 \nvDash K_a(p_b \land p_c)$ for all two-valued models corresponding to $\C^{X-b}_3$.

\section{Axiom System \texorpdfstring{$\Sfive^{\isdef}$}{S5\unichar{"22C8}}}
\label{axiom_sys}

The main goal of this paper is to provide a sound and complete axiomatization for the notion of validity from Definition~\ref{def:valid}.  Since pure simplicial models correspond to Kripke models with equivalence accessibility relations, their logic is an extension of multimodal~$\Sfive$,  the epistemic logic where individual knowledge has both positive and negative introspection (the extension deals with the locality of propositional variables). The susceptibility of agents to crashes can be reasonably expected to weaken the logic somewhat but one could still expect some subsystem of $\Sfive$.

As we have discovered, the impact of dead agents and undefined formulae is quite drastic and, interestingly, affects the core reasoning, both modal and propositional,  rather than the additional epistemic properties. Although these results are from~\cite{vDitKuz22arXiv}, we find it instructive to repeat the proofs found there as it provides insight to the reader on why these axioms/rules  are invalid in their ``usual'' form.

Firstly, the \emph{normality axiom}~\textbf{K}  at the basis of all normal modal logics is not valid even for the above weaker notion of validity:

\begin{lem}[Failure of normality \cite{vDitKuz22arXiv}]
\label{ex:counterK}
In the presence of at least three agents, 
\[
 \nvDash K_a(\phi \to \psi) \to (K_a \phi \to K_a \psi).
\]
\end{lem}
\begin{figure}[ht]
\center
\begin{tikzpicture}[round/.style={circle,fill=white,inner sep=1}]
\fill[fill=gray!25!white] (2,0) -- (4,0) -- (3,1.71) -- cycle;
\node[round] (b1) at (0,0) {$0_b$};
\node[round] (b0) at (4,0) {$1_b$};
\node[round] (c1) at (3,1.71) {$1_c$};
\node[round] (a0) at (2,0) {$0_a$};
\node (f1) at (3,.65) {$Y$};

\draw[-] (b1) -- node[above] {$X$} (a0);
\draw[-] (a0) -- (b0);
\draw[-] (b0) -- (c1);
\draw[-] (a0) -- (c1);

\node(c) at (-1,.85) {$\C_{\mathbf{K}}:$};
\end{tikzpicture}
\caption{Counterexample for the unrestricted axiom \textbf{K}}
\label{fig:counterK}
\end{figure}
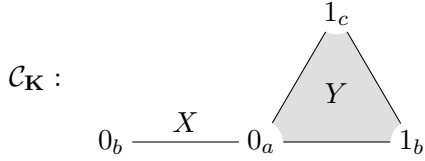
\begin{proof}
Consider model~$\C_{\mathbf{K}}$ with agents $a$, $b$, and $c$ in Figure~\ref{fig:counterK}. Let $\phi = p_c$ and $\psi = p_b$. Model~$\C_{\mathbf{K}}$ has only two facets, $X$ of dimension 1 and $Y$ of the maximal dimension~2. The  instance 
\begin{equation}
\label{eq:badK}
K_a(p_c \to  p_b) \to (K_a p_c \to K_a  p_b)
\end{equation}
 of axiom~\textbf{K} is defined in both facets because all formulae are defined in $Y$, which is $a$-adjacent to all facets. In addition, $p_c$ is true wherever it is defined, making $K_a p_c$  true in both facets. Similarly, $p_c \to p_b$ is true wherever it is defined, \ie~wherever both $b$ and $c$ are alive, because $p_b$~is only false in $X$ where $c$ is dead. This ensures that $K_a( p_c \to p_b)$ is true in both facets. At the same time, $\C_{\mathbf{K}}, X \vDash \neg p_b$ and, hence, $K_a p_b$ is false in both facets. Thus, not only is  instance~\eqref{eq:badK} of axiom \textbf{K} falsified in this model, but, in fact, its negation is valid in~$\C_{\mathbf{K}}$.
\end{proof}

Perhaps, even more surprisingly, even the \emph{modus ponens} \textbf{MP} rule, a staple of most logics, does not generally preserve validity. In~\cite{vDitKuz22arXiv} (though not in~\cite{vDit21WoLLIC}), we presented a counterexample to this effect that involved four agents. Here we offer a mild improvement and demonstrate that already for three agents, \textbf{MP} does not preserve validity:
\begin{lem}[Failure of modus ponens]
\label{ex:noMP}
In the presence of at least three agents, there exist formulae~$\phi$~and~$\psi$ such that 
\[
\vDash \phi \to \psi \quad \text{and} \quad \vDash \phi,\qquad \text{but}\qquad \nvDash \psi.
\]
\end{lem}
\begin{proof}
Our counterexample to {modus ponens}  preserving validity uses
 $\phi = \top_a \land  \top_b $ and $\psi = \M_c \M_a\M_b p_c \lor \M_c \M_b \M_a \lnot p_c$. 
We will prove the validity of $\phi$ and of $\phi\to\psi$, \ie~that
	\begin{align}
\label{eq:one}
\vDash\quad&
 {\top_a \land  \top_b},
\\
\label{eq:two}
\vDash\quad&
 \top_a \land   \top_b 
\imp 
\M_c \M_a\M_b p_c \lor \M_c \M_b \M_a \lnot p_c.
\end{align}
but that applying modus ponens to these valid formulae does not produce a valid formula. For the latter part,
consider simplicial model~$\C_{\mathbf{MP}}$ with agents $a$, $b$, and $c$ in Figure~\ref{fig:counterMP}, which has six facets, all  of non-maximal dimension~$1$. 
To show that  $\nvDash \M_c \M_a\M_b p_c \lor \M_c \M_b \M_a \lnot p_c$, we demonstrate
\begin{equation*}
\label{eq:three}
\C_{\textbf{MP}}, X\vDash\quad \neg\left(
\M_c \M_a\M_b p_c \lor \M_c \M_b \M_a \lnot p_c\right).
\end{equation*}

It remains to prove the former part, \ie~validities~\eqref{eq:one}~and~\eqref{eq:two}.
The validity of formula~$\phi$ in~\eqref{eq:one} immediately follows from our earlier discussion that local boolean constants \mbox{$\top_e \ce p_e \lor \neg p_e$} cannot be falsified. Let us show that formula $\phi \to \psi$ in \eqref{eq:two} is valid. Assume that $\phi \to \psi$ is defined in a face $Q$ of a simplicial model  $\C$. This necessitates that agents~$a$~and~$b$ be alive due to $\top_a$ and $\top_b$, as well as agent~$c$ being alive due to  modalities~$\M_c$. In other words, $\phi \to \psi$ is only defined in $Q$ if $Q$ is a face with $a$, $b$, and $c$ all being alive. 
In particular, variable~$p_c$ is defined, \ie~is either true or false in $Q$. Since face $Q$ is $a$-adjacent, $b$-adjacent, and \mbox{$c$-adjacent} to itself, if $p_c$ is true in~$Q$, so is $\M_c \M_a\M_b p_c$, while if $\lnot p_c$ is true in~$Q$, so is~$\M_c \M_b \M_a \lnot p_c$. This completes the proof that $\top_a \land   \top_b
\imp 
\M_c \M_a\M_b p_c \lor \M_c \M_b \M_a \lnot p_c$ is true whenever defined.

Let us now show that $\psi = \M_c \M_a\M_b p_c \lor \M_c \M_b \M_a \lnot p_c$ is invalidated in  model~$\C_{\mathbf{MP}}$. 
For the second disjunct $\M_c \M_b \M_a \lnot p_c$ evaluated in facet $X$, out of three faces $c$-adjacent to~$X$, \ie~$X$, $Y$, and their common vertex, agent $b$ is alive only in $Y$, hence, the definability/truth of $\M_c \M_b \M_a \lnot p_c$ in $X$ is fully determined by definability/truth of  $\M_b \M_a \lnot p_c$ in $Y$.  
Out of three faces $b$-adjacent to $Y$, \ie~$Y$, $V$, and their common vertex, agent $a$ is alive only in~$V$, hence, the definability/truth of~$\M_b \M_a \lnot p_c$ in $Y$ is fully determined by definability/truth of~$\M_a \lnot p_c$ in~$V$. 
Out of three faces $a$-adjacent to $V$, \ie~$V$, $U$, and their common vertex, $\lnot p_c$ is defined only in $U$, hence, $ \M_a \lnot p_c$ is defined at $V$ and its truth is fully determined by the truth of~$\lnot p_c$ in~$U$. 
Since $\lnot p_c$ is false in $U$, it follows that $\M_a \lnot p_c$ is false in $V$, that $\M_b \M_a \lnot p_c$ is defined and false in $Y$, and that $\M_c \M_b \M_a \lnot p_c$ is defined and false in $X$. 
The argument for the first disjunct is similar, except the chain of faces there is $X$, followed by $X$, followed by~$W$, and finally followed by $Z$.
 In summary, $\M_c \M_a\M_b p_c \lor \M_c \M_b \M_a \lnot p_c$ is defined but false in~$X$.
 
 \begin{figure}[t]
\center
\begin{tikzpicture}[round/.style={circle,fill=white,inner sep=1}]
\node[round] (b1) at (0,0) {$a$};
\node[round] (b0) at (4,0) {$b$};
\node[round] (d0) at (6,0) {$0_c$};
\node[round] (d1) at (-2,0) {$1_c$};
\node[round] (c1) at (3,1.71) {$a$};
\node[round] (lc1) at (1,1.71) {$b$};
\node[round] (a0) at (2,0) {$c$};
\node (f1) at (2.3,.9) {$X$};
\node (f2) at (1.7,.9) {$Y$};
\node (f3) at (3.75,.9) {$W$};
\node (f4) at (0.35,.9) {$V$};
\node (f5) at (5,.2) {$Z$};
\node (f5) at (-1,.2) {$U$};

\draw[-] (b1) -- (d1);
\draw[-] (b1) -- (lc1);
\draw[-] (a0) -- (lc1);
\draw[-] (d0) -- (b0);
\draw[-] (b0) -- (c1);
\draw[-] (a0) -- (c1);

\node(c) at (-3,.85) {$\C_{\textbf{MP}}:$};
\end{tikzpicture}
\caption{Counterexample for the unrestricted {modus ponens}}
\label{fig:counterMP}
\end{figure}
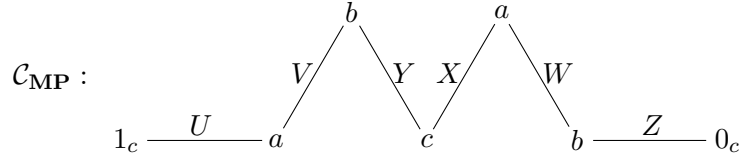

 It is interesting to note that  implication~\eqref{eq:two} is nowhere defined in $\C_{\mathbf{MP}}$.
\end{proof}

Intuitively, the failure of both \textbf{K} and \textbf{MP} can be attributed to the fact that different instances of modality in impure simplicial models may have different scope in terms of adjacent faces, which prevents the usual first-order reasoning about the adjacency. One way to amend this is by imposing  definability restrictions so as to ensure the homogeneity of the scope of modal operators. 

\begin{defi}[Axioms system $\Sfive^{\isdef}$]
\label{fig:sfiveimp}
The axioms and rules of axiom system\/~$\Sfive^{\isdef}$ can be found below (non-standard axioms are in the right column):
\[
\begin{array}{cl|cll}
\mathbf{Taut} & 
	\text{prop. tautologies} 
	& 
	\mathbf{L} & 
	K_a p_a \lor K_a \lnot p_a
	   \\
\mathbf{T} &  
	K_a \phi \imp \phi & \mathbf{K}^{\bowtie} &
	K_a(\phi \imp \psi) \imp (K_a \phi \imp K_a \psi) & \text{ where }\psi, K_a \phi  \impdef \phi\\
	\mathbf{4} & 
	K_a \phi \imp K_a K_a \phi 
	&
	\mathbf{K\M} & 
	K_a(\phi \imp \psi) \imp (K_a \phi \imp \M_a \psi) \\
\mathbf{5} & 
	\lnot K_a \phi \imp K_a\lnot K_a\phi 
	&\mathbf{MP}^{\bowtie}  & 
	\text{ from } \phi \imp \psi \text{ and } \phi, \text{ infer } \psi
	&\text{ where } \psi \impdef \phi\\ 
\mathbf{N} & 
	\text{ from } \phi, \text{ infer } K_a \phi
\\
\end{array}
\]
\end{defi}

Thus,  axiom system $\Sfive^{\isdef}$ is obtained by 
\begin{itemize}
\item taking all axioms and rules of multimodal $\Sfive$   that are valid for impure simplicial models (see the left column in Definition~\ref{fig:sfiveimp}), 
\item adopting the locality axiom \textbf{L} used for pure simplicial models, 
\item restricting \textbf{K} and \textbf{MP} by imposing definability conditions, and
\item  adding axiom $\mathbf{K\M}$.
\end{itemize}
This last addition is needed for the completeness proof. Intuitively, it can be explained by the dual universal--existential nature of  modalities $K_a$. In impure semantics, knowledge of~$\phi$ means both truth in \emph{all} adjacent faces where $\phi$ is defined (universal aspect) and truth in \emph{at~least~one} adjacent face (existential aspect). Just like its unrestricted variant~\textbf{K}, axiom~$\mathbf{K}^{\bowtie}$ deals only with the universal aspect. The new axiom $\mathbf{K\M}$ is needed to take care of the existential one.

\begin{rem}
Logic~$\Sfive^{\isdef}$ does not satisfy the property of uniform substitution.
 Note that agent $a$'s local propositional variable $p_a$ in axiom~$\mathbf{L}$ cannot be replaced with propositional variable $q_b$ of some other agent $b$, let alone with an arbitrary formula. Indeed, $K_a q_b \vee K_a \neg q_b$ for $a \ne b$ is not generally valid, and assuming its validity would go against the ideology of distributed systems where agents are autonomous and do not generally know each other's local state. This lack of closure with respect to substitution is yet another of the multiple ways logic~$\Sfive^{\isdef}$ fails to be a normal modal logic, albeit in this case inherited from the logic of pure simplicial complexes.
\end{rem}

\section{Soundness}
\label{soundness}

\begin{defi}[Derivability in $\Sfive^{\isdef}$]
\label{def:deriv}
A formula $\phi$ is \emph{derivable} in\/~$\Sfive^{\isdef}$, written\/ $\vdash \phi$, if{f} there is a finite sequence $\psi_1,\dots,\psi_k$ of formulae such that $\psi_k = \phi$ and each formula in the sequence is either an instance of one of the axioms in Definition~\ref{fig:sfiveimp} or follows from earlier formulae in the sequence by one of the rules from Definition~\ref{fig:sfiveimp}. A formula $\phi$ is derivable in\/~$\Sfive^{\isdef}$ from a set\/~$\Gamma$ of formulae, written\/ $\Gamma \vdash \phi$, if{f} either\/ $\vdash \phi$ or\/ $\vdash \bigwedge \Gamma_0 \imp \phi$ for some finite non-empty subset\/ $\Gamma_0 \subseteq \Gamma$.
\end{defi}

\begin{rem}
The proviso that $\Gamma_0 \ne \varnothing$ and the special case for  derivations from the empty set are necessary due to the lack of boolean constants in the language, which makes it problematic to interpret the empty conjunction.
\end{rem}

\begin{restatable}[Soundness]{thm}{soundimp}
\label{th:soundimp}
Axiom system\/~$\Sfive^{\isdef}$ is sound with respect to~impure simplicial models:
\[
\vdash \phi\qquad \text{implies}\qquad \vDash \phi.
\]
\end{restatable}
\begin{proof}
The proof by induction on an $\Sfive^{\isdef}$-derivation is standard: one shows that all axioms are valid and all rules preserve validity. It has already been shown in~\cite{vDitKuz22arXiv} that the locality axiom~$\mathbf{L}$ and all axioms from the left column of Definition~\ref{fig:sfiveimp} are valid and that the necessitation rule~$\mathbf{N}$ preserves validity. For  readers' convenience, we reproduce these proofs in Appendix~\ref{b}. Here we provide  the argument for the remaining two axioms and  rule:
\begin{itemize}
\item Axiom $\mathbf{K\M}$ is valid.
	Consider an arbitrary simplicial model $\C = (C,\chi, \ell)$ and its face~\mbox{$X \in C$}. If one of formulae $K_a(\phi \imp \psi)$ and $K_a \phi$ is undefined,  axiom $\mathbf{K\M}$ is undefined. If both formulae are defined but one of them is false,   $\mathbf{K\M}$ is true. Finally,  assume that \mbox{$\C,X \vDash K_a(\phi \imp \psi)$} and $\C,X \vDash K_a \phi$, 
	Then,  $\C, X \isdef K_a(\phi \imp \psi)$ by Lemma~\ref{lem:threevalues}\ref{clause:threevalues_one}, \ie~$\C, Y \isdef \phi \imp \psi$ for some face~$Y$ such that $a \in \chi(X \cap Y)$, where $\C, Y \isdef \phi$ and $\C, Y \isdef \psi$ by the definition of~$\isdef$. Since both $\phi \to \psi$ and $\phi$ are defined in $Y$, they must be true by our assumptions, \ie~$\C,Y \vDash \phi \to \psi$ and 
 $\C,Y \vDash \phi$. Hence, $\C,Y \vDash \psi$ and $\C,X \vDash \M_a \psi$. 
	We conclude that in this final  case again $\C,X \vDash K_a(\phi \imp \psi) \imp (K_a \phi \imp \M_a \psi)$.
	
\item Axiom $\mathbf{K}^{\isdef}$ is valid, provided $\psi, K_a \phi \impdef \phi$. 
From now on, we omit all trivial cases and assume that $\C,X \vDash K_a(\phi \imp \psi)$ and $\C,X \vDash K_a \phi$. As already shown for axiom $\mathbf{K\M}$, the former assumption implies that
 $\psi$ is defined in some face $a$-adjacent to face $X$. It remains to show that $\psi$ is true in all such $a$-adjacent faces where it is defined. Consider any face~$Z$ such that $a \in \chi(X \cap Z)$ and  $\C, Z \isdef \psi$. Before we can apply our assumptions to~$Z$, we have to make sure that  both $\phi \to \psi$ and~$\phi$ are defined in $Z$ using the definability provision of~$\mathbf{K}^{\isdef}$.
  $\C,X \vDash K_a \phi$ means that $\phi$ is defined in at least one face $Y$ $a$-adjacent to~$X$. But $Y$ being $a$-adjacent to~$X$ is equivalent to $Y$ being $a$-adjacent to $Z$ because $a \in \chi(X \cap Z)$. Hence, $\C,Z \isdef K_a \phi$. Therefore,  provision $\psi, K_a \phi \impdef \phi$ ensures that $\C, Z \isdef \phi$ and, consequently, $\C, Z \isdef \phi \to \psi$. Now we can use the assumptions to obtain  $\C,Z \vDash\phi $ and  $\C,Z \vDash \phi \to \psi$, meaning that $\C,Z \vDash\psi $. This concludes the proof that $\C,X \vDash K_a \psi$ and, hence, of the validity of $\mathbf{K}^{\isdef}$.

\item	Rule $\mathbf{MP}^{\isdef}$ preserves validity, provided  $\psi \impdef \phi$. We assume that $\vDash \phi \to \psi$ and $\vDash \phi$ and need to show that $\psi$ is true whenever defined. The definability provision means that $\phi$~is defined whenever $\psi$ is, hence, so is $\phi \to \psi$. Their validity means that $\phi$ and $\phi\to\psi$~are true whenever $\psi$ is defined, and from that we conclude that $\psi$ is true whenever it is defined.
	\qedhere
	\end{itemize}
\end{proof}

\begin{rem}
While all the standard epistemic axioms are valid, the meaning of this validity sometimes defies the usual intuitions. For instance, axiom \textbf{T} is often called the \emph{factivity of knowledge} and normally states that whatever is known must be true. This interpretation is, however, too strong for impure simplicial models and is not entailed by the validity~$\vDash K_a \phi \to \phi$. Indeed, for model~$\C_{\mathbf{K}}$ from Figure~\ref{fig:counterK}, $\C_{\mathbf{K}}, X \vDash K_a p_c$ even though formula $p_c$ is undefined rather than true in facet~$X$. In  our three-valued semantics, the validity of~\textbf{T} means  `known facts cannot be false'  rather than   `known facts must be true.' In this context, perhaps, it is better described as the \emph{veracity of knowledge}. This is why we may still see logic~$\Sfive^{\isdef}$ as an epistemic logic (see \cite{vDitKuz22arXiv} for an extensive discussion on the epistemic interpretation).
\end{rem}

In preparation for the completeness proof,  we now illustrate how the impurity and its associated definability restrictions affect standard modal (and propositional) reasoning techniques. 

\begin{lem}[Hypothetical syllogism]
	\label{lem:condHS}
	 If $\phi \imp \rho\impdef \psi$, the following rule is derivable
	\[
	\displaystyle\trule{\vdash \phi \imp \psi}{\vdash \psi \imp \rho}{ \vdash \phi \imp \rho}.
	\]
\end{lem}
\begin{proof}
\hfill
	\begin{enumerate}[1.]
		\item $(\phi \imp \psi) \imp \bigl((\psi \imp \rho) \imp (\phi \imp \rho)\bigr)$ 
		\hfill tautology
		\item $\phi \imp \psi$
		\hfill derivable by assumption
		\item $\psi \imp \rho$
		\hfill derivable by assumption
		\item $(\psi \imp \rho) \imp (\phi \imp \rho)$
		\hfill \textbf{MP}$^{\bowtie}$(1.,2.) as $(\psi \imp \rho) \imp (\phi \imp \rho) \impdef \phi \imp \psi$
		\item $\phi \imp \rho$
		\hfill \textbf{MP}$^{\bowtie}$(3.,4.) as $\phi \imp \rho \impdef \psi$ and, hence, $\phi \imp \rho \impdef \psi \imp \rho$
	\end{enumerate}
\end{proof}

\begin{restatable}[Contraposition]{lem}{contrapos}
	\label{lem:contrapos}
	The following rule is derivable:
	\[
	\displaystyle\orule{\vdash \phi \imp \psi}{ \vdash \lnot \psi \imp \lnot \phi}.
	\] 
\end{restatable}
\begin{proof}
See Appendix~\ref{b}.
\end{proof}

\begin{restatable}[Replacement for $K_a$]{lem}{kamonot}
	\label{lem:syl5}
	If   $\psi , K_a \phi \impdef \phi$, the following rule is derivable:
	\[\displaystyle \orule{\vdash\phi \imp \psi}{\vdash K_a \phi \imp K_a \psi}.\]
\end{restatable}
\begin{proof}
See Appendix~\ref{b}.
\end{proof}

\begin{restatable}[Partial normality of $K_a$ with respect to conjunction]{lem}{kaconj}
	\label{lem:ConjOfKaImpKaOfConj}
	\begin{align*}
		&\vdash K_a \phi \land K_a \psi  \to K_a(\phi \land \psi)
		\\
	&\nvdash    K_a(\phi \land \psi) \to K_a \phi \land K_a \psi.
	\end{align*}
\end{restatable}
\begin{proof}
See Appendix~\ref{b}.
\end{proof}

As desired, in this logic, agents know the truth value of their local variables:
\begin{restatable}[Knowledge of local variables]{lem}{localknowledge}
	\label{lem:locality_positive}
	\label{lem:locality_negative}
	\begin{align*}
	&\vdash \phantom{\lnot}p_a \imp K_a p_a
	\\
	&\vdash \lnot p_a \imp K_a \lnot p_a.
	\end{align*}
\end{restatable}
\begin{proof}
\hfill
	\begin{enumerate}[1.]
		\item $K_a p_a \to p_a$ \hfill axiom \textbf{T}
		\item $\lnot p_a \imp \lnot K_a p_a$
		\hfill Lemma~\ref{lem:contrapos}(1.)
		\item $K_a p_a \lor K_a \lnot p_a$
		\hfill axiom \textbf{L}
		\item $(K_a p_a \lor K_a \lnot p_a) \imp (\lnot K_a p_a \imp K_a \lnot p_a)$
		\hfill tautology
		\item $\lnot K_a p_a \imp K_a \lnot p_a$
		\hfill MP$^{\isdef}$(3.,4.) as $\lnot K_a p_a \imp K_a \lnot p_a \impdef K_a p_a \lor K_a \lnot p_a$
		\item $\lnot p_a \imp  K_a \lnot p_a$
		\hfill Lemma~\ref{lem:condHS}(2.,5.) as $\lnot p_a \imp K_a \lnot p_a \impdef \lnot K_a p_a$
	\end{enumerate}
	This derives  the second formula. The first one can be derived analogously.
\end{proof}

\begin{restatable}[Deduction theorem]{thm}{deductth}
\[
\G, \phi \vdash \psi 
\qquad \Longrightarrow \qquad
\G \vdash \phi \imp \psi.
\]
\end{restatable}
\begin{proof}
See Appendix~\ref{b}.
\end{proof}

The following three lemmatas are  used in the completeness proof:

\begin{restatable}{lem}{addextraK}
\label{lem:add_extra_K}
 For $m\geq 1$,
\[
\vdash\quad K_a \psi_1 \land \cdots \land K_a \psi_m\quad \imp\quad K_a(K_a \psi_1 \land \cdots \land K_a \psi_m).
\]
\end{restatable}
\begin{proof}
See Appendix~\ref{b}.
\end{proof}

\begin{restatable}{lem}{derforcons}
	\label{lem:derivationForConsistency}
	For $m \geq 1$, the following rule is derivable:
	\[\displaystyle \orule{ \vdash \lnot (\phi \land K_a \psi_1 \land \cdots \land K_a \psi_m)}{ \vdash \lnot (K_a \phi \land K_a \psi_1 \land \cdots \land K_a \psi_m)}.\]
\end{restatable}
\begin{proof}
See Appendix~\ref{b}.
\end{proof}

\begin{restatable}{lem}{derforconsneg}
	\label{lem:derivationForConsistencyNeg}
		For $m \geq 1$, the following rule is derivable:
	\[\displaystyle\orule{ \vdash \lnot (\lnot \phi \land K_a \psi_1 \land \cdots \land K_a \psi_m)}{ \vdash \lnot (\lnot K_a  \phi \land K_a \psi_1 \land \cdots \land K_a \psi_m)}.
	\]
\end{restatable}
\begin{proof}
See Appendix~\ref{b}.
\end{proof}

\section{Completeness via the Canonical Simplicial Model}
\label{completeness}

In this section, we present the novel construction of the canonical \emph{simplicial} model. Our main challenge  was that 
  the standard canonical \emph{Kripke} model construction,  based on maximal consistent sets, builds the canonical model out of Kripke worlds representing global states, whereas a simplicial model is built from vertices representing  local states, \ie~agents' view of the world.

Before giving the technical account of the necessary changes and a proof of completeness, we explain the chosen construction on the structural level. To do that, we briefly recall how  standard canonical Kripke models are constructed  and used to prove  completeness. A~canonical model $\mathcal{M}^C$ for a given logic can be viewed as a universal countermodel, in the sense that any formula that can be refuted in some world~$w$ of some Kripke model $\mathcal{M}$ can also be refuted in a world of the canonical model $\mathcal{M}^C$. 
To this effect, for any world $w$ of any Kripke model $\mathcal{M}$, there is a world in the canonical model that assigns exactly the same truth values as~$w$, which is achieved by using the set  $\Gamma_{\mathcal{M},w} \ce \{\phi \mid \mathcal{M}, w \vDash \phi\}$ of formulae that are true at world~$w$ of Kripke model $\mathcal{M}$ as a world in the canonical model. In fact, the domain of the canonical model consists of all sets~$\Gamma_{\mathcal{M},w}$ over all models $\mathcal{M}$ and worlds~$w$. Note that in the binary boolean semantics the set  $\Gamma_{\mathcal{M},w}$ of formulae that are true in~$w$ provides a complete description of all truth values in~$w$: the set of formulae false in~$w$ is~$\overline{\Gamma_{\mathcal{M},w}}$, the complement of the set of true formulae. 

The completeness proof is achieved by combining these syntactic objects (sets of formulae)~--- via an appropriate valuation and accessibility relations ---  into a Kripke model in such a way that semantic truth can be reduced to syntactic membership. This reduction is termed the \emph{Truth Lemma} and states that a formula $\phi$ is true at a world $\Gamma$ of the canonical model if{f} $\phi \in \Gamma$. Again, the syntactic criterion for falsity at a world of the canonical model easily follows: 
\begin{equation}
\label{eq:neg_twin}
\text{a formula $\phi$ is false at a world $\Gamma$}
\qquad \Longleftrightarrow \qquad 
\phi \notin \Gamma 
\qquad \Longleftrightarrow \qquad 
\lnot \phi \in \Gamma.
\end{equation}

Of course, building a model from all worlds of the class of all Kripke models would have been unwieldy, especially given that all $\Gamma_{\mathcal{M},w}$ are subsets of the set of all formulae, making the set of worlds of the canonical model a subset of the power set of the set of all formulae. Instead, the canonical model is built of \emph{maximal consistent sets} that are syntactically consistent with respect to the given axiomatization, whereas the condition of maximality ensures that a set is maximally consistent if{f} it is equal to $\Gamma_{\mathcal{M},w}$ for some Kripke model $\mathcal{M}$ and world $w$.

Our construction of the canonical \emph{simplicial} model below adapts these ideas to the semantics of impure simplicial complexes. At first, it might seem that, unlike in the binary case,  knowing the set of formulae that are true would not be sufficient to represent the distribution of our three truth values. How would we be able to distinguish formulae that are false from those that are undefined? Fortunately, it is possible to do by separating the two descriptions of falsity in~\eqref{eq:neg_twin}, which were equivalent in the binary semantics. In our three-valued case, we will formulate the truth lemma to associate true formulae with formulae in $\Gamma$, false formulae with formulae \emph{whose negation} is in $\Gamma$, with all remaining formulae, \ie~formulae $\phi$ such that both $\phi \notin \Gamma$ and $\lnot \phi \notin \Gamma$,  undefined. This enables us to continue building the canonical model out of sets of formulae rather than pairs of sets or some such.\looseness=-1

The next hurdle is that such sets of formulae represent a possible distribution of three truth values in a global state, whereas we need to find a syntactic representation of local states for each agent $a$. There are two immediate candidates for formulae belonging to~$a$: propositional variables $p_a \in P_a$ and knowledge assertions $K_a \phi$. We will show that the latter are sufficient to describe the local state of $a$ and, hence, each $a$-colored vertex in our canonical simplicial model will be a set $K_a\G$ of formulae of the form $K_a \phi$. Accordingly, each face in the canonical simplicial model will be a finite set $X=\{K_{a_1}\G, K_{a_2}\G, \dots, K_{a_m}\G\}$ of sets~$K_{a_i}\G$ of knowledge assertions for distinct agents $a_1,a_2, \dots, a_m$ alive in a global state~$\G$, and the whole canonical simplicial model will be a set of such finite sets $X$, supplied with appropriate valuation and coloring functions.

In the logic of impure simplicial models, it is possible to restore  the global state $\Gamma$ from the collection 
\[
X_\G=\{K_{a_1}\G, K_{a_2}\G, \dots, K_{a_m}\G\}
\]
of local states for all agents $a_1,a_2,\dots,a_m$ alive in~$\G$.

There is one final piece of the puzzle before we begin the formal part of the proof. As discussed above, in the boolean case, \emph{maximal} consistent sets represent all possible truth sets across all worlds of all Kripke models. We need to do the same for all facets of all impure simplicial models, and it is clear that the absolute maximality condition is too strong for that. In Kripke models, adding any formula to a truth set of a given world would make the set inconsistent. If, however, agent $a$ is dead in a given facet,  then adding a propositional variable~$p_a \in P_a$ to the set of formulae that are true in this facet need not lead to inconsistency. There may well be another facet where $a$ is alive and some previously undefined formulae, including~$p_a$, are true there. (It is also possible that adding $p_a$ does lead to inconsistency if, \eg,~$K_b \lnot p_a$ is present in the given truth set.) 
In our canonical model construction, we replace maximal consistent sets with maximal consistent subsets of definability-complete sets.
\begin{defi}[Definability completeness]
A set\/ $\Xi\ne \varnothing$ is \emph{definability complete} if{f} it contains all its definability consequences, \ie~if $\phi \in \Xi$ whenever\/ $\Xi \impdef \phi$.
\end{defi}

To streamline proofs and avoid the reliance on $\Xi$ as an extra parameter, we use a more direct description of maximal consistent subsets of such $\Xi$'s as \emph{sets that determine all their definability consequences}.
\begin{defi}[Being determined]
	\label{def:notation}
 	A formula $\phi$ is \emph{determined} by a set\/~$\Gamma$ of formulae, written
$
\phi \Sin \G$,
{if{f}} either $\phi \in \G$ or\/  $\lnot \phi \in \G$.
Accordingly, a formula is \emph{not determined}, written~$\phi \notSin \G$, if{f}\/ $\Gamma\cap \{\phi, \lnot \phi\} = \varnothing$.
\end{defi}

\begin{defi}[Consistent sets]
	\label{def:consistent}
	A  set\/~$\G$ of formulae is \emph{inconsistent} if{f}\/ $\vdash \lnot\bigwedge \G_0$ for some finite\/ $\varnothing \ne \G_0 \subseteq \G$. Otherwise,\/ $\G$~is \emph{consistent}. In particular,\/ $\varnothing$~is consistent.
\end{defi} 
\begin{rem}
The non-standard requirement  $\G_0 \ne \varnothing$ is due to the lack of  boolean constants.
\end{rem}

\begin{defi}[Definability-maximal consistent (DMC) sets]
	\label{def:locallyMax}
	 A consistent set\/ $\G\ne \varnothing$ is \emph{definability maximal  (a DMC set)} if{f} it determines all its definability consequences:
	 	 \begin{equation}
	 \label{eq:def_DMC}
	 \G \impdef \phi \qquad\Longrightarrow\qquad \phi \Sin \G.
	 \end{equation}
\end{defi}

The following facts are standard  for classical propositional reasoning but need to be reaffirmed here due to our weaker version of modus ponens.
\begin{restatable}[Monotonicity of~$\vdash$]{lem}{supsetincons}
\label{lem:add_form}
For formulae $\phi$ and $\psi$ and a finite set\/ $\G \ne \varnothing$ of formulae:
\begin{enumerate}[(a)]
\item\label{lem:add_form_neg}
If\/  $\vdash \lnot\bigwedge \G$, then\/ $\vdash \lnot\left(\phi \land \bigwedge \G\right)$.
\item\label{lem:add_form_deriv}
If\/ $\vdash \bigwedge \G \to \psi$, then\/ $\vdash \phi \land \bigwedge \G \to \psi$.
\end{enumerate} 
\end{restatable}
\begin{proof}
See Appendix~\ref{b}.
\end{proof}

Let us justify the term \emph{definability maximal} by showing that DMC sets are exactly maximal consistent subsets of definability-complete sets:

\begin{lem}
\label{lem:DMC_sets_determ}
For any  consistent set\/ $\G\ne \varnothing$,  the following statements are equivalent:
\begin{enumerate}[(i)]
\item\label{DMC_one} $\G$ is a DMC set; 
\item\label{DMC_two}  $\G$~is maximal among consistent subsets of some definability-complete set\/~$\Xi$.
\end{enumerate} 
\end{lem}
\begin{proof}

\ref{DMC_one} $\Longrightarrow$ \ref{DMC_two}. 
Let $\G$ be  a DMC set, \ie~let  \eqref{eq:def_DMC}~hold. We start by showing that  set~$\Xi \ce \{\phi \mid \G \impdef \phi\}$ is definability complete. Indeed, if $\Xi \impdef \psi$, then  $\G \impdef \psi$ by Lemma~\ref{lem:eq_definabilities}\ref{lem:eq_def_trans}, and $\psi \in \Xi$ by construction. Since $\G \impdef \psi$ for any $\psi \in \G$, clearly  $\G \subseteq \Xi$, in particular, $\Xi \ne \varnothing$. It remains to show that $\G$ is maximal among consistent subsets of $\Xi$. Suppose, towards a contradiction, that $\G \cup \{\psi\}$ were consistent for some $\psi \in \Xi \setminus \G$. Then $\psi \notin \G$ but $\G \impdef \psi$, which, due to~\eqref{eq:def_DMC}, would imply that $\lnot \psi \in \G$. But then $\vdash \lnot (\psi \land \lnot \psi)$ for $\{\psi, \lnot \psi \} \subseteq \G \cup \{\psi\}$, contradicting the supposition of consistency for this set. The contradiction shows that $\G$ is maximal among consistent subsets of the definability-complete set $\Xi$.

\ref{DMC_two} $\Longrightarrow$ \ref{DMC_one}. Let $\Xi$ be a definability-complete set, $\G$ be maximal among  consistent subsets of $\Xi$, and let $\G \impdef \phi$. Then  there is some finite subset $\G_0\subseteq \G$  such that $\phi$ is defined whenever all formulae from $\G_0$ are. Since $\G_0\subseteq \G \subseteq \Xi$,  we have $\Xi \impdef \phi$. Since $\phi \eqdef \lnot \phi$, also  $\Xi \impdef \lnot \phi$. Due to the definability completeness of $\Xi$, it means that $\{\phi, \lnot \phi\} \subseteq \Xi$. 
Suppose towards a contradiction that $\phi \notin \G$ and $\lnot \phi \notin \G$.
Due to the maximality of $\G$ among consistent subsets of $\Xi$, it would mean that both $\G \cup \{\phi\}$ and $\G \cup \{\lnot \phi\}$ would be inconsistent. In view of Lemma~\ref{lem:add_form}\ref{lem:add_form_neg}, without loss of generality, 
$
\vdash \lnot \bigl(\phi \land\bigwedge \G_0 \land \bigwedge \G_1\bigr)$
and
$\vdash \lnot \bigl(\lnot \phi \land \bigwedge \G_0 \land \bigwedge \G_2\bigr)$
 for some non-empty finite  subsets $\G_1, \G_2 \subseteq \G$.
 The following is a propositional tautology:
\[
\vdash\quad  \lnot\bigl(\phi \land \bigwedge \G_0 \land\bigwedge \G_1\bigr)\quad \to\quad \left(\lnot\bigl(\lnot \phi \land\bigwedge \G_0 \land \bigwedge \G_2\bigr) \to \lnot \bigl(\bigwedge \G_0 \land \bigwedge\G_1 \land  \bigwedge \G_2\bigr)\right).
\]
Since it is clear that
\[
\lnot\bigl(\lnot \phi \land\bigwedge \G_0 \land \bigwedge \G_2\bigr) \to \lnot \bigl(\bigwedge \G_0 \land \bigwedge\G_1 \land  \bigwedge \G_2\bigr) 
\qquad\impdef\qquad \lnot\bigl(\phi \land \bigwedge \G_0 \land\bigwedge \G_1\bigr)
,\]
it would follow by $\mathbf{MP}^{\isdef}$ that
\[
\vdash\quad \lnot\bigl(\lnot \phi \land\bigwedge \G_0 \land \bigwedge \G_2\bigr)\quad \to\quad \lnot \bigl(\bigwedge \G_0 \land \bigwedge\G_1 \land  \bigwedge \G_2\bigr).
\]
Note, additionally, that  
\[
\lnot \bigl(\bigwedge \G_0 \land \bigwedge\G_1 \land  \bigwedge \G_2\bigr)  
\qquad\impdef\qquad
\lnot\bigl(\lnot \phi \land\bigwedge \G_0 \land \bigwedge \G_2\bigr)
.\]
Indeed, whenever $\lnot \bigl(\bigwedge \G_0 \land \bigwedge\G_1 \land  \bigwedge \G_2\bigr)$ is defined,  all formulae from $\G_2$ and $\G_0$ are defined, the latter implying that $\phi$ is defined. Hence, $\lnot\bigl(\lnot \phi \land\bigwedge \G_0 \land \bigwedge \G_2\bigr)$ is defined. Thus, another application of $\mathbf{MP}^{\isdef}$  would yield
$
\vdash \lnot \bigl(\bigwedge \G_0 \land \bigwedge\G_1 \land  \bigwedge \G_2\bigr).
$ 
But $\G_0 \cup \G_1 \cup \G_2$ is a finite subset of~$\G$, creating a contradiction with the assumption that $\G$ is consistent. This contradiction shows that $\G$ determines $\phi$. Since $\phi$ was chosen arbitrarily, $\G$~determines all its definability consequences, \ie~is a DMC set.
\end{proof}

We start our formal presentation by showing that DMC sets represent all possible truth sets across all faces of all simplicial models. This is slightly more than representing all facets, \ie~all global states, but creates no additional problems.

\begin{lem}
	\label{lem:canmod_from_Faces}
	For any simplicial model $\C=(C,\chi,\ell)$ and its face $X \in C$, the truth set
	\[
	\G_{\C,X} \ce \{ \phi  \mid \C,X \vDash \phi \}
	\]  
	is a DMC set.
\end{lem}
\begin{proof}
It is clear that $\G_{\C,X}$ cannot be empty. Hence, to show that $\G_{\C,X}$ is consistent, consider any non-empty finite subset $\G_0 \subseteq \G_{\C,X}$. We have $\C,X \vDash \bigwedge \G_0$ by construction. Therefore, $\C,X \nvDash \lnot \bigwedge \G_0$ and $\nvdash \lnot\bigwedge \G_0$ by the soundness of $\Sfive^{\bowtie}$ (Theorem~\ref{th:soundimp}). Since $\G_0$ was chosen arbitrarily, we have established that $\G_{\C,X}$ is consistent. 

	It remains to show that $\G_{\C,X}$ is definability maximal (DM). 
	Assume that $\G_{\C,X} \impdef \phi$, \ie~there is a finite subset $\G_0 \subseteq \G_{\C,X}$ such that $\phi$ is defined whenever all formulae from~$\G_0$~are. Since, by construction, $\C,X \vDash \psi$ for all $\psi \in \G_{\C,X}$, it follows from Lemma~\ref{lem:threevalues}\ref{clause:threevalues_one} that $\C,X \isdef \psi$ for all $\psi \in \G_0$. Thus,  $\C,X \isdef \phi$, meaning that either $\C,X \vDash \phi$ or $\C,X \vDash \lnot \phi$. By construction, either $\phi$ or $\lnot \phi$ belongs to $\G_{\C,X}$, \ie~$\phi \Sin \G_{\C,X}$.
\end{proof}

Properties of DMC sets differ slightly from those  of maximal consistent sets:
\begin{lem}[Properties of DMC sets]
\label{lem:LMCclosed}
For DMC sets\/~$\G$ and\/~$\D$, an agent~$a$, propositional variable $p_a\in P_a$, and formulae $\phi$, $\phi_1$, $\phi_2$, and $\psi$:
\begin{enumerate}[(a)]
\item\label{follows}
$\G \impdef \phi$\quad and\quad  $\G \vdash \phi$ \qquad $\Longrightarrow$ \qquad $\phi \in \G$.
\item\label{notboth}
$\G \impdef \phi$\qquad $\Longrightarrow$ \qquad exactly one of\quad  $\phi \in \G$\quad and\quad $\lnot \phi\in \G$\quad holds.
\item\label{closeconj}
$\phi_1 \land \phi_2 \in \G$\qquad $\Longleftrightarrow$ \qquad both\quad $\phi_1\in\G$\quad and\quad $\phi_2\in\G$.
\item
\label{lem:LCMclosedstarWrtKa}
	$K_a\G  \neq \varnothing$\quad and\quad  $\phi \Sin \G$ \qquad $\Longrightarrow$ \qquad  $K_a\phi \Sin \G$. 
\item
\label{atomsthere}
$p_a \in \G$ iff $K_a p_a \in \G$, and\/ $\lnot p_a \in \G$ iff $K_a \lnot p_a \in \G$. Hence,  $p_a \Sin \G$ implies $K_a \G \ne \varnothing$.
\item	
\label{lem:addingPhiDoesNotChangeKa}
	\label{lem:NeverAddKa}
$ K_a \G, \phi \impdef K_a\psi$\quad and\quad $K_a\phi \Sin \G$ \qquad$\Longrightarrow$\qquad $K_a\psi \Sin \G$.

\item
\label{lem:neg_neg_also_carries}
$K_a \G \subseteq K_a \D$ \quad and\quad 
$\lnot K_a \phi \in \G $
\qquad$\Longrightarrow$\qquad $\lnot K_a \phi \in \D$.

\item
\label{clause:ifKthennonot}
$K_a \phi \in \G$\qquad$\Longrightarrow$\qquad $\lnot \phi \notin\G$\quad and\quad $K_a\G \cup\{\phi\}$ is consistent.
\item
\label{clause:ifnotKthennophi}
$\lnot K_a \phi \in \G$
\qquad$\Longrightarrow$\qquad
$K_a \G \cup\{\lnot \phi\}$ is consistent.
\end{enumerate} 
\end{lem}
\begin{proof}
\hfill
\begin{enumerate}[(a)]
\item 
Since $\G \vdash \phi$, we have $\vdash \bigwedge\G_0 \to \phi$ for some finite $\varnothing \ne \G_0 \subseteq \G$. 
Given tautology \mbox{$\vdash \bigl(\bigwedge\G_0 \to \phi\bigr) \imp \lnot \bigl(\lnot \phi \land \bigwedge \G_0\bigr)$},  it follows by \textbf{MP}$^{\isdef}$ that $\vdash \lnot \bigl(\lnot \phi \land \bigwedge \G_0\bigr)$  because $ \lnot \bigl(\lnot \phi \land \bigwedge \G_0\bigr) \impdef \bigwedge\G_0 \to \phi$. Thus, $\lnot \phi \in \G$ would contradict $\G$'s consistency. On the other hand, either $\phi \in \G$ or $\lnot \phi \in \G$ because $\G$ is DM and $\G \impdef \phi$.  Thus, $\phi \in \G$.
\item
It follows from $\vdash \lnot (\phi \land \lnot \phi)$ and $\G$ being  DM. 
\item
It follows by clause~\ref{follows} since  for $i=1,2$
$
\phi_1 \land \phi_2  \vdash \phi_i$,
$\phi_1 \land \phi_2  \impdef \phi_i$, 
$\phi_1, \phi_2 \vdash \phi_1 \land \phi_2$, 
and
$\phi_1, \phi_2 \impdef \phi_1 \land \phi_2$.
\item
By Lemma~\ref{lem:eq_definabilities}\ref{lem:if_phi_then_ka_phi}, $\G \impdef K_a \phi$. 
Hence, $K_a\phi \Sin \G$ 
because $\G$ is DM.

\item
By Lemma~\ref{lem:eq_definabilities}\ref{lem:eq_atom_K}, since $p_a \eqdef \lnot p_a$, it follows that $p_a \eqdef K_a p_a$ and $\lnot p_a \eqdef K_a \lnot p_a$. By Lemma~\ref{lem:locality_positive}, $p_a \vdash K_a p_a$ and $\lnot p_a \vdash K_a \lnot p_a$. By axiom \textbf{T}, $K_a p_a \vdash p_a$ and $K_a \lnot p_a \vdash \lnot p_a$. Hence, the first two if{f} claims follow from clause~\ref{follows}. The last claim follows from them.

\item
 By Lemma~\ref{lem:eq_definabilities}\ref{def_raise_ka}, $K_a \G,K_a\phi\impdef K_a\psi$, hence, also $K_a \G,\lnot K_a\phi\impdef K_a\psi$. 
Thus, however  $\G$~determines $K_a \phi$, we know that  $\G \impdef K_a\psi$. Now $K_a \psi\Sin \G$ follows from $\G$ being DM.

\item By  clause~\ref{follows},  since  $\lnot K_a \phi \impdef K_a\lnot K_a \phi$  and additionally $\lnot K_a \phi \vdash K_a\lnot K_a \phi$ due to axiom~\textbf{5}, it follows that, if $\lnot K_a \phi \in \G$, then  $K_a \lnot K_a \phi \in K_a \G \subseteq K_a \D$. Similarly, $K_a\lnot K_a \phi \impdef \lnot K_a \phi$ by Lemma~\ref{lem:eq_definabilities}\ref{def:eucl} and $K_a \lnot K_a \phi \vdash \lnot K_a \phi$ by axiom~\textbf{T}. Hence, $\lnot K_a \phi \in \D$ by clause~\ref{follows}.

\item
Since $\lnot (K_a \phi \land \lnot \phi)$ is equidefinable and equivalent to $K_a \phi \to \phi$, \ie~to axiom~\textbf{T}, having $\lnot \phi \in\G$ would have contradicted the consistency of $\G$. (Note that, unlike in traditional modal logic, this does not necessarily mean $\phi \in \G$.)
	We prove the consistency of~$K_a \G \cup \{\phi\}$  by contradiction. Were it not,   we would have $K_a \psi_1,\dots, K_a \psi_m \in K_a \G\subseteq \G$  such that (without loss of generality in view of Lemma~\ref{lem:add_form}\ref{lem:add_form_neg}) $\vdash \lnot (\phi \land  K_a \psi_1 \land \dots \land K_a \psi_m)$. 
	A contradiction $\vdash \lnot (K_a \phi \land K_a \psi_1 \land \dots \land K_a \psi_m)$ to the consistency of $\G$ would follow by Lemma~\ref{lem:derivationForConsistency}.
\item
The argument here is similar to the preceding clause, except that Lemma~\ref{lem:derivationForConsistencyNeg} is used in place of Lemma~\ref{lem:derivationForConsistency} and instead of assuming $K_a \G \ne \varnothing$, the non-emptiness is inferred from $K_a \lnot K_a \phi \in  \G$, which is obtained by the same argument as in clause~\ref{lem:neg_neg_also_carries}.
\qedhere
\end{enumerate} 
\end{proof}

\begin{lem}[Lindenbaum Lemma]
\label{lem:lind}
Any consistent set\/ $\G\ne \varnothing$ can be extended to a DMC  set\/ $\D \supseteq \G$ such that\/ $\G \impdef \psi$ for all $\psi \in \D$.
\end{lem}
\begin{proof}
Let $\psi_0,\psi_1,\dots$ be an  enumeration of the set $\G^{\impdef} \colonequals \{\psi  \mid \G \impdef \psi\}$  of all definability consequences of $\G$. Starting from $\D_0\ce\G$, add formulae $\psi_i$ one by one, unless doing so would cause inconsistency: $\D_{i+1} \ce \D_i \cup \{\psi_i\}$ if that set is consistent or $\D_{i+1} \ce \D_i$ otherwise.  Define $\D \ce \cup_{i=0}^\infty \D_i \subseteq \G^{\impdef}$, which is consistent for the same reason as in the standard Lindenbaum construction that enumerates all formulae: were $\D$ inconsistent, one of~$\D_i$ would have been.
We now show that  $\D$ is maximal among consistent subsets of $\G^{\impdef}$. Indeed, any~$\psi \in \G^{\impdef} \setminus \D$ is one of $\psi_i$'s and was not added to $\D_{i}$ due to $\D_{i} \cup\{\psi\}$ being inconsistent. Hence, $\D \cup \{\psi\}$ is also inconsistent. Since $\G^{\impdef}$ is definability complete by  Lemma~\ref{lem:eq_definabilities}\ref{lem:eq_def_trans}, it follows from Lemma~\ref{lem:DMC_sets_determ} that $\D$ is a DMC set.
\end{proof}

\begin{defi}[Simplicial canonical model]
	\label{def:simpIm}
	\label{def:simpcanonic}
	Let\/ $\LMC$  denote the collection of all DMC sets. For\/  $\G \in \LMC$, we define
	\[
		   X_{\G} \colonequals \{K_a \G \mid a \in A, K_a\G \ne \varnothing \} .
	\]
	The \emph{simplicial canonical model} $\canmod = (C^c, \chi^c, \ell^c)$ is defined as follows:
	\begin{itemize}
		\item $C^c \ce \{X   \mid (\exists \G \in \LMC)\, (\varnothing \ne X \subseteq X_{\G} )\}$.
		 Accordingly, its set of vertices can be computed as $\VV^c= \bigcup_{X \in C^c} X =  \{K_a \G \mid \G \in \LMC, a \in A, K_a\G \ne \varnothing\}$.

		\item For $K_a \G \in \VV^c$, we define $\chi^c(K_a \G) \ce a $. 
		\item For $K_a \G \in \VV^c$, we define $\ell^c(K_a \G) \ce \{p_a \in P_a \mid K_a p_a \in K_a \G \}$. 
	\end{itemize}
\end{defi}

In other words, given $\G \in \LMC$, non-empty sets $K_a \G$ are  vertices of the canonical model~$\canmod$ and sets $X_\G$  are, in most cases,  its facets (but not necessarily, see Section~\ref{reflections}). The color of vertex~$K_a \G$ must obviously be $a$, and the local variables that are true there are  those known by this agent $a$. 
\begin{rem}
Unlike sets $\G \in \LMC$, sets $K_a \G$ are not DMC sets. For instance, neither $p_a \in K_a \G$ nor $\lnot p_a \in K_a \G$ even though either $K_a p_a \in K_a \G$ or  $K_a \lnot p_a \in K_a \G$ and $p_a$ is a definability consequence of either.
\end{rem}

\begin{restatable}[Correctness]{lem}{cancor}
\label{lem:cancor}
The object $\canmod$ constructed in Definition~\ref{def:simpcanonic}  is a simplicial model.
\end{restatable}
\begin{proof}
The proof is straightforward and can be found in Appendix~\ref{b}.
\end{proof}

The Truth Lemma has to be adapted to the logic of impure simplicial models in two ways:\looseness=-1
\begin{enumerate}[(i)]
\item formulae may be undefined, in addition to being true or false, necessitating an extra clause~\eqref{eq:TL_undef} in the Truth Lemma;
\item the Truth Lemma is only proved for faces of the form $X_\G$ because their proper subfaces~$X \subsetneq X_\G$ do not generally correspond to DMC sets.
\end{enumerate} 

\begin{lem}[Truth Lemma]
	\label{lem:truth}
	Let $\canmod$ be the canonical simplicial model constructed in Definition~\ref{def:simpcanonic}. Then, 
	for any DMC set\/  $\G \in \LMC$ and any formula $\phi$,
	\begin{align}
	\label{eq:TL_undef}
		\phi \notSin  \G \qquad &\Longrightarrow \qquad\canmod, X_{\G} \ndef \phi,
		\\
	\label{eq:TL_true}
		\phi \in \G\qquad &\Longrightarrow \qquad\canmod, X_{\G} \vDash \phi,
		\\
	\label{eq:TL_false}
		\lnot \phi \in \G\qquad &\Longrightarrow \qquad\canmod, X_{\G} \vDash \lnot\phi.
	\end{align}
\end{lem}
\begin{proof}
	The proof is by induction on the construction of $\phi$. Note that $X_\G \in C^c$.\footnote{See the proof of Lemma~\ref{lem:cancor} in Appendix~\ref{b}.}
	\begin{description}
		\item[For $\phi = p_a$]\hfill
		\begin{itemize}
		\item[\eqref{eq:TL_undef}]
			Assume $p_a \notSin \G$. Then $\G \not\impdef p_a$ because $\G$ is DM , hence,  $K_a\G  = \varnothing$ by Lemma~\ref{lem:eq_definabilities}\ref{from_K_to_atom}. Therefore, $K_a\G \notin X_{\G}$, $a  \notin \chi^c(X_{\G})$, and $\canmod, X_{\G} \ndef p_a$.
\item[\eqref{eq:TL_true}]
			Assume $p_a \in \G$. Then,  $K_a p_a \in K_a \G$ by Lemma~\ref{lem:LMCclosed}\ref{atomsthere}. Since $K_a\G\ne \varnothing$, we have $K_a\G \in \VV^c$, and  $p_a \in \ell^c(K_a\G) \subseteq\ell^c(X_{\G})$. Therefore,  $\canmod, X_{\G} \vDash p_a$.
\item[\eqref{eq:TL_false}]
			Assume $\lnot p_a \in \G$.  Then,  $K_a \lnot p_a \in K_a \G$ by Lemma~\ref{lem:LMCclosed}\ref{atomsthere}. Again, $K_a\G\ne \varnothing$ implies $K_a\G \in \VV^c$ and $a \in \chi^c(X_\G)$. Therefore, $\canmod, X_\G \isdef p_a$. On the other hand,  $p_a \notin \G$ by the consistency of $\G$ implying $K_a p_a \notin K_a \G$ by Lemma~\ref{lem:LMCclosed}\ref{atomsthere}. Hence, $p_a \notin \ell^c(X_{\G})$ meaning $\canmod, X_\G \nvDash p_a$. Thus, in summary, $\canmod, X_\G \models \lnot p_a$.
	\end{itemize}
		\item[For $\phi = \lnot \psi$]
		\hfill
		\begin{itemize}
		\item[\eqref{eq:TL_undef}]
Assume $\lnot \psi \notSin  \G$. Then $\psi \notSin \G$  since $\G$ is DM and $\psi \impdef \lnot \psi$. Then $\canmod, X_{\G} \ndef \psi$  by  induction hypothesis~\eqref{eq:TL_undef} for $\psi$. Accordingly, $\canmod, X_{\G} \ndef \lnot \psi$.
\item[\eqref{eq:TL_true}]			
			Assume $\lnot \psi \in \G$. Then, $\canmod, X_\G \models \lnot \psi$ by  induction hypothesis \eqref{eq:TL_false} for $\psi$.
\item[\eqref{eq:TL_false}]			 
			 Assume $\lnot\lnot \psi \in \G$. Then $\psi \in \G$ by Lemma~\ref{lem:LMCclosed}\ref{follows} because  $\lnot\lnot \psi \impdef \psi$ and $\lnot \lnot \psi \vdash \psi$. Thus, $\canmod, X_{\G} \vDash \psi$  by  induction hypothesis~\eqref{eq:TL_true} for $\psi$. Accordingly, $\canmod, X_{\G} \vDash \lnot \lnot \psi$.
	\end{itemize}
		\item[For $\phi = \psi_1 \land \psi_2$]\hfill
		\begin{itemize}
		\item[\eqref{eq:TL_undef}]
Assume $\psi_1 \land \psi_2 \notSin  \G$. Then $\psi_i \notSin \G$ for at least one of the conjuncts  since $\G$~is~DM and $\widetilde{\psi_1}, \widetilde{\psi_2} \impdef \psi_1 \land \psi_2 $ for any $\widetilde{\psi_j} \in \{\psi_j, \lnot \psi_j\}$. Then $\canmod, X_{\G} \ndef \psi_i$  by induction hypothesis~\eqref{eq:TL_undef} for this $\psi_i$. Accordingly, $\canmod, X_{\G} \ndef \psi_1 \land \psi_2$.
\item[\eqref{eq:TL_true}]			
			Assume $\psi_1 \land \psi_2 \in \G$. Then, $\psi_i \in \G$ for both conjuncts by Lemma~\ref{lem:LMCclosed}\ref{closeconj}. Thus, 
			$\canmod, X_\G \models  \psi_i$  by  induction hypothesis \eqref{eq:TL_true} for both $\psi_i$. Accordingly, $\canmod, X_\G \vDash\psi_1 \land \psi_2$. 
\item[\eqref{eq:TL_false}]			 
			 Assume $\lnot (\psi_1 \land \psi_2)\in \G$. Then $\psi_j \Sin \G$  since $\G$ is DM and  $\lnot (\psi_1 \land \psi_2)\impdef \psi_j$ for both conjuncts.  Hence, $\canmod, X_\G \isdef\psi_j$  by Lemma~\ref{lem:threevalues}\ref{clause:threevalues_one}--\ref{clause:threevalues_two} and  induction hypotheses~\mbox{\eqref{eq:TL_true}--\eqref{eq:TL_false}} for both $\psi_j$. Thus, $\canmod, X_{\G} \isdef \psi_1 \land \psi_2$. On the other hand, $\psi_1\land \psi_2\notin \G$ by the consistency of $\G$, hence, $\psi_i \notin \G$ for at least one of $\psi_i$			 
			  by Lemma~\ref{lem:LMCclosed}\ref{closeconj}. It follows that $\lnot \psi_i \in \G$ and $\canmod, X_\G \nvDash  \psi_i$ by  induction hypothesis~\eqref{eq:TL_false} for this $\psi_i$. Hence, $\canmod, X_\G \nvDash  \psi_1 \land \psi_2$ and, overall, $\canmod, X_\G \vDash  \lnot (\psi_1 \land \psi_2)$.
	\end{itemize}
	\item[For $\phi=K_a \psi$]\hfill
	\begin{itemize}
	\item[\eqref{eq:TL_undef}] Assume that $K_a \psi \notSin \G$.	
		\begin{itemize}
			\item If $K_a\G  = \varnothing$, then  $K_a\G \notin X_{\G}$. Therefore, $a  \notin \chi^c(X_{\G})$, and there can be no~$Y \in C^c$ such that $a \in \chi^c(X_{\G} \cap Y)$ and $\canmod, Y \isdef \psi$. In this case, we have $\canmod, X_{\G} \ndef K_a\psi$.
			\item If $K_a\G  \neq \varnothing$, then $K_a\G \in X_{\G}$. Consider any $Y \in C^c$ such that $a \in \chi^c(X_{\G} \cap Y)$. This implies that $K_a\G \in Y$. By Definition~\ref{def:simpcanonic}, $\varnothing \ne Y \subseteq X_{\D}$ for some $\D \in \LMC$ and, hence, $K_a\G \in X_{\D}$. Therefore, $K_a\D =K_a\G $, and it follows that $K_a\D \neq \varnothing$, meaning that  $a \in \chi^c(X_\G \cap X_\D)$. Since  $K_a \psi \Sin \D$ if{f} $K_a \psi \Sin \G$ by Lemma~\ref{lem:LMCclosed}\ref{lem:neg_neg_also_carries}, it follows that  $K_a \psi \notSin \D$.
			By Lemma~\ref{lem:LMCclosed}\ref{lem:LCMclosedstarWrtKa},  $\psi \notSin \D$ and, by induction hypothesis~\eqref{eq:TL_undef}, $\canmod, X_{\D} \ndef \psi$. From Lemma~\ref{lem:defMonotonicity}\ref{monot_defin}, it follows that $\canmod, Y \ndef \psi$. Since face $Y$, which is $a$-adjacent to~$X_\G$, was chosen arbitrarily,  $\canmod, X_{\G} \ndef K_a\psi$.
		\end{itemize}
	\item[\eqref{eq:TL_true}]	
		Assume  that $K_a \psi \in \G$. Hence, $K_a \G \ne \varnothing$, $K_a \G \in X_\G$, and $a \in \chi^c(X_\G)$. By Corollary~\ref{cor:simple_def}\ref{eq:trueK}, for $\C^c, X_\G \vDash K_a \psi$, it is sufficient to show that:
		\begin{enumerate}[(i)]
			\item\label{enum:Kall} $\canmod, Y \isdef \psi$ implies $\canmod, Y \vDash \psi$ for all $Y \in C^c$ such that $a \in \chi^c(X_{\G}\cap Y)$, and
			\item\label{enum:Ksome} $\canmod, Y \vDash \psi$ for some $Y \in C^c$ with $a \in \chi^c(X_{\G}\cap Y)$.
		\end{enumerate} 
		To show~\ref{enum:Kall}, consider any $Y \in C^c$ such that $a \in \chi^c(X_{\G}\cap Y)$ and $\canmod, Y \isdef \psi$. It follows from $a \in \chi^c(X_{\G}\cap Y)$ that $K_a\G \in Y$. By Definition~\ref{def:simpcanonic}, $\varnothing \ne Y \subseteq X_{\D}$ for some $\D \in \LMC$. Then, on the one hand, $K_a\G \in X_{\D}$ which implies that $K_a \D = K_a\G$, and, on the other hand, $\canmod, X_{\D} \isdef \psi$ by Lemma~\ref{lem:defMonotonicity}\ref{monot_defin}. It follows from  induction hypothesis~\eqref{eq:TL_undef} that $\psi \Sin \D$. Since $K_a \D = K_a\G$ we have $K_a \psi \in \D$ and, hence, $\lnot \psi \notin \D$ by Lemma~\ref{lem:LMCclosed}\ref{clause:ifKthennonot}. Therefore, $\psi \in \D$. Thus, $\canmod, X_{\D} \vDash \psi$ by  induction hypothesis~\eqref{eq:TL_true}, and $\canmod, Y \vDash \psi$ by Lemma~\ref{lem:defMonotonicity}\ref{monot_back}. We have shown that $\psi$ is true in all faces $a$-adjacent  to $X_\G$ where it is defined.
		
		To show~\ref{enum:Ksome}, consider $K_a \G \cup \{\psi \}$, which is consistent by Lemma~\ref{lem:LMCclosed}\ref{clause:ifKthennonot}. 
		By Lindenbaum Lemma~\ref{lem:lind}, there exists a DMC set $\D \supseteq  K_a \G \cup \{\psi \}$ such that $K_a \G, \psi  \impdef \psi$ for all~$\psi \in \D$.  
		We have $K_a \G \subseteq K_a \D$ by construction. To prove the opposite inclusion, consider any~$K_a \theta \in \D$.
		Since $K_a\G, \psi \impdef K_a \theta$, by Lemma~\ref{lem:eq_definabilities}\ref{def_raise_ka}  $K_a \G, K_a \psi  \impdef K_a\theta$, and, hence, $\G \impdef K_a \theta$. Then  $K_a \theta \Sin \G$ because $\G$ is DM. 
		Were $\lnot K_a \theta \in \G$, we would have had  $\lnot K_a \theta \in \D$ by Lemma~\ref{lem:LMCclosed}\ref{lem:neg_neg_also_carries}, making $\D$ inconsistent. It follows that $K_a \theta \in \G$ for all~$K_a \theta \in \D$. Thus, $K_a \G = K_a \D$ and $a \in \chi^c(X_{\G} \cap X_{\D})$. It remains to note that, by induction hypothesis~\eqref{eq:TL_true} for $\psi$, we have $\canmod, X_{\D} \vDash \psi$, and $X_\D$ serves as a face  $a$-adjacent  to $X_\G$ where $\psi$ is true.
\item[\eqref{eq:TL_false}]		
		Assume $\lnot K_a \psi \in \G$. By Corollary~\ref{cor:simple_def}\ref{eq:falseK}, for $\C^c, X_\G \vDash \lnot K_a \psi$, it is sufficient to show that
		$\canmod, Y \vDash \lnot \psi$ for some $Y \in C^c$ with $a \in \chi^c(X_{\G}\cap Y)$.
		Consider $K_a \G \cup \{\lnot \psi \}$, which is consistent by Lemma~\ref{lem:LMCclosed}\ref{clause:ifnotKthennophi}. By Lindenbaum Lemma~\ref{lem:lind}, there exists a DMC~set \mbox{$\D \supseteq  K_a \G \cup \{\lnot \psi \}$} such that $K_a \G, \lnot \psi  \impdef \psi$ for all $\psi \in \D$. By construction, $K_a \G \subseteq K_a \D$. To prove the opposite inclusion, consider any $K_a \theta \in \D$. By Lemma~\ref{lem:eq_definabilities}\ref{def_raise_ka_neg},  \mbox{$K_a \G, \lnot K_a \psi  \impdef K_a\theta$}, and, hence, $\G \impdef K_a \theta$. Then  $K_a \theta \Sin \G$ because $\G$~is~DM. 
		Were $\lnot K_a \theta \in \G$,  this would have meant that $\lnot K_a \theta \in \D$ by Lemma~\ref{lem:LMCclosed}\ref{lem:neg_neg_also_carries},  making $\D$~inconsistent. It follows that $K_a \theta \in \G$ for all $K_a \theta \in \D$. Thus, $K_a \G = K_a \D$ and $a \in \chi^c(X_{\G} \cap X_{\D})$. It remains to note that, by  induction hypothesis~\eqref{eq:TL_false}   for $\psi$, we have $\canmod, X_{\D} \vDash \lnot \psi$, and $X_\D$~serves as a face  $a$-adjacent to~$X_\G$ where $\psi$~is false.
\qedhere		
\end{itemize}
	\end{description}
\end{proof}

\begin{thm}[Completeness]
	Axiomatization\/ $\Sfive^{\isdef}$ is complete with respect to~impure simplicial models:
	\[
	\vDash \phi
	\qquad\Longrightarrow\qquad
	\vdash \phi.
	\]
\end{thm}
\begin{proof}
	We prove the contraposition. Suppose $ \nvdash \phi$ for a formula $\phi$. Then $\{\lnot \phi\}$ is a consistent set. By Lindenbaum Lemma~\ref{lem:lind}, $\{\lnot \phi\}$ is a subset of some DMC set $\G$. By Truth Lemma~\ref{lem:truth}, $\canmod, X_{\G} \vDash \lnot \phi$. Thus, $\phi$ is not valid.\end{proof}

\section{Simplicial Canonical Model Reflections}
\label{reflections}

In most similar papers, the completeness theorem would be the happy ending. The deed is done, leaving only  the conclusions section to summarize the results and possibly set up a sequel. Here, however, we would like to take a moment to reflect on the structure of the canonical simplicial model~$\canmod$ constructed in Definition~\ref{def:simpcanonic} and how well it matches the class of all impure simplicial models. As with the normality axiom \textbf{K} and modus ponens rule \textbf{MP}, there are surprising departures from the properties we are used to expect, which makes it worthwhile to dispel some possible misconceptions. 

One structural property we have largely ignored until now was the distinction between facets, \ie~maximal faces, and faces that are not maximal. As mentioned earlier, it is facets that correspond to global states, \ie~worlds in Kripke models, whereas other faces are just intermediate constructs. One can even argue that these non-facet faces are just artifacts of simplicial complex structure with no significance of their own. It was Lemma~\ref{lem:val_facet} that gave us  license  to ignore this distinction. If the logic does not depend on whether we look only at global states  or additionally allow their partial shadows, there seemed little reason to bother about separating these shadows away. 

As a result, we constructed  canonical simplicial model~$\canmod$ to faithfully represent all faces of all simplicial models, which, of course, includes all facets of all simplicial models. But we did not make an effort to ensure that those facets are represented by facets in~$\canmod$. One might expect the canonicity to take care of that. We now show that such an expectation would be misplaced, \ie~that there are facets represented in~$\canmod$ by non-maximal faces.

\begin{exa}
\label{ex:canon_merge}
In impure simplicial complexes, the situation when two agents are unsure whether the third agent is dead or alive is typically represented by 
simplicial models of the form $\C_{\mathrm{Xmas}} = (C,\chi, \ell)$, which  we are going to call \emph{Christmas cracker models},
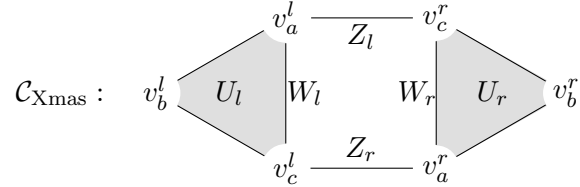
\begin{figure}[H]%
	\centering	 
\begin{tikzpicture}[round/.style={circle,fill=white,inner sep=1},scale=0.5]
\fill[fill=gray!25!white] (-2, 2) -- (-5.46, 0) -- (-2, -2) ;
\fill[fill=gray!25!white] (2, -2) -- (5.56, 0) -- (2, 2) ;

\node[round]  (a0) at (-2, 2) {$v_a^l$};
\node[round]  (c1) at (2, 2) {$v_c^r$};
\node[round]  (a1) at (2, -2) {$v_a^r$};
\node[round]  (c0) at (-2, -2) {$v_c^l$};
\node[round]  (b0) at (-5.46, 0) {$v_b^l$};
\node[round]  (b1) at (5.46, 0) {$v_b^r$};

\draw [-] (a0) -- (c1) node at (1.5,0) {$W_r$};
\draw [-] (a0) -- (b0);
\draw [-] (a0) -- (c0) node at (0,-1.5) {$Z_r$};
\draw [-] (c0) -- (b0);
\draw [-] (c1) -- (b1);
\draw [-] (c0) -- (a1) node at (-1.5,0) {$W_l$};
\draw [-] (a1) -- (b1);
\draw [-] (a1) -- (c1) node at (0,1.5) {$Z_l$};

\node (f1) at (-3.5,0) {$U_l$};
\node (f2) at (3.5,0) {$U_r$};

\node(c) at (-8,0) {$\C_{\mathrm{Xmas}}:$};
\end{tikzpicture}
	\caption{Christmas cracker models}
	\label{fit:Xmas}
\end{figure}
\noindent
for three agents $a$, $b$, and $c$  with
	\begin{itemize}
		\item  six \textbf{vertices} $\VV= \{v_a^l,v_a^r, v_b^l,v_b^r,v_c^l,v_c^r \}$,
		\item four \textbf{facets} $\FF(C)= \bigl\{\{v_a^l,v_b^l,v_c^l\},\quad \{v_a^r,v_b^r,v_c^r\},\quad \{v_a^l,v_c^r\},\quad \{v_c^l,v_a^r\} \bigr\}$,
		\item a \textbf{chromatic map} $\chi  \colon v^{*}_i \mapsto i$, which maps each vertex to its subscript, and
		\item a \textbf{valuation} $\ell$ that must satisfy the following conditions:
		\begin{equation}
		\label{eq:Xmas_val_cond}
		\ell(v_a^l)=\ell(v_a^r),
		\qquad
		\ell(v_b^l)=\ell(v_b^r),
		\qquad
		\ell(v_c^l)=\ell(v_c^r).		
		\end{equation}
	\end{itemize}
Let us use $U_l$ and $U_r$ for 2-dimensional facets (triangles), $Z_l$ and $Z_r$ for 1-dimensional facets~(edges), and $W_l\subset U_l$ and $W_r \subset U_r$ for the $ac$-colored 1-dimensional faces of the triangles (see Figure~\ref{fit:Xmas}). It immediately follows from the rotational symmetry that, omitting the name of the simplicial model and using 
\begin{equation}
\label{eq:Xmas}
\Gamma_{V} \ce \{ \phi \mid \C_{\mathrm{Xmas}}, V \vDash \phi\}
\end{equation}
for the truth sets of its faces, we get
$
\Gamma_{U_l} = \Gamma_{U_r}$,
$\Gamma_{Z_l} = \Gamma_{Z_r}$,
$\Gamma_{W_l} = \Gamma_{W_r}$,~etc.
(recall that the set of formulae that are true in a face also provides a complete description of formulae  false and  undefined in the same face). These sets are DMC sets according to Lemma~\ref{lem:canmod_from_Faces}. 

Accordingly, in the canonical model~$\canmod$, facets $Z_l$ and $Z_r$ are represented by the same canonical simplex $X_{ \Gamma_{Z_l}}=\{K_a \Gamma_{Z_l},  K_c \Gamma_{Z_l}\}$ of dimension 1. By the same token, faces~$W_l$~and~$W_r$~are represented there by  canonical simplex $\{K_a \Gamma_{U_l},  K_c \Gamma_{U_l}\}$ of dimension~1, which is a canonical face within canonical facet $X_{ \Gamma_{U_l}}=\{K_a \Gamma_{U_l},  K_b \Gamma_{U_l}, K_c \Gamma_{U_l}\}$ of the maximal dimension~2. We will show that
$
\{K_a \Gamma_{Z_l},  K_c \Gamma_{Z_l}\} 
=
\{K_a \Gamma_{U_l},  K_c \Gamma_{U_l}\}
$,
\ie~both facets~$Z_l$~and~$Z_r$ and non-maximal faces $W_l$ and $W_r$ are represented by the same  canonical face that is not a facet.

$K_a \Gamma_{Z_l} = K_a \Gamma_{U_l}$ follows directly  from Definitions~\ref{def:defin} and~\ref{def:truth} because  $Z_l$ and $U_l$ share  $a$-colored vertex $v_a^l$. 
To show that $K_c \Gamma_{Z_l} = K_c \Gamma_{U_l}$, it is sufficient to note that the truth values of formulae $K_c \phi$ in $Z_l$ are determined by truth values of $\phi$ in all simplices $c$-adjacent to~$Z_l$, \ie~in
$
\{v_c^r\},  Z_l,  W_r,  \{v_c^r,v_b^r\}$, and  $U_r$,
whereas the same truth values in $U_l$ are determined by truth values of $\phi$ in all simplices $c$-adjacent to $U_l$, \ie~in
$
\{v_c^l\},  Z_r,  W_l, \{v_c^l,v_b^l\}$, and $U_l$.
As already noted earlier, for every of these five pairs of corresponding simplices, truth values are identical. Hence, the truth value of each $K_c \phi$ in $Z_l$ and $U_l$ is also the same.
\end{exa}

This example makes it clear that, while all facets are represented in the canonical model, they may not be represented by a canonical facet. The same example demonstrates that definability-maximal consistent sets need  not be maximal consistent in the usual sense but only with respect to subsets of some definability-complete set.
\begin{lem}
\label{lem:DMCnotMAX}
There exist  DMC sets\/ $\Gamma$ and\/~$\Delta$ such that\/ $\Gamma \subsetneq \Delta$.
\end{lem}
\begin{proof}
It is sufficient to show that $\Gamma_{Z_l} \subsetneq \Gamma_{U_l}$ in Example~\ref{ex:canon_merge}.  
Note that all formulae are defined in $U_l$ because all agents are alive there. We show by mutual induction that any formula that is true (false) in $Z_l$ is true (respectively false) in $U_l$, \ie~that
\[
\C, Z_l  \vDash \phi 
\quad\Longrightarrow \quad
\C, U_l \vDash \phi
\qquad\text{and}\qquad
\C, Z_l  \vDash \lnot \phi 
\quad\Longrightarrow \quad
\C, U_l \vDash \lnot \phi.
\]
For propositional variables $\phi = p_a \in P_a$, this is trivial as truth values for them are determined by~$\ell(v^l_a)$. For $\phi = p_c \in P_c$, it follows from~\eqref{eq:Xmas_val_cond}. No $p_b\in P_b$ is defined in $Z_l$. 

For $\phi = \lnot \psi$, if $\lnot \psi$ is true in $Z_l$, then $\psi$ is false there, hence, $\psi$ is false in $U_l$ by the induction hypothesis, making $\lnot \psi$ true in $U_l$. 
The argument for  false $\lnot \psi$  in $Z_l$ is analogous.

For $\phi = \psi_1 \land \psi_2$, if $\psi_1 \land \psi_2$ is true in $Z_l$, then both  $\psi_1$ and $\psi_2$ are true there, hence, both are true in $U_l$ by the induction hypothesis, and $ \psi_1 \land \psi_2$ is true in $U_l$. If $\psi_1 \land \psi_2$ is false in~$Z_l$, then at least one of the conjuncts is false there. Hence, this conjunct is false in $U_l$ by the induction hypothesis. Since $\psi_1 \land \psi_2$ is defined in $U_l$, it can only be false.

No $K_b \psi$ is defined in $Z_l$. For $\phi = K_a \psi$, if it is true in $Z_l$, then $K_a \psi \in K_a \Gamma_{Z_l} $ according to~\eqref{eq:Xmas}. It was shown in Example~\ref{ex:canon_merge} that $K_a \Gamma_{Z_l}= K_a \Gamma_{U_l}$, \ie~$K_a \psi$ is also true in~$U_l$ by~\eqref{eq:Xmas}. If $K_a \psi$ is false, then $K_a \psi \notin K_a \Gamma_{Z_l} = K_a \Gamma_{U_l}$. Since $K_a \psi$ is defined in $U_l$, it must be false. The case of $\phi = K_c \psi$ is analogous.

Thus, $\Gamma_{Z_l} \subseteq \Gamma_{U_l}$, and it remains to note that either $p_b$ or $\lnot p_b$ must belong to~$\Gamma_{U_l}$ (depending on  valuation $\ell$), but neither belongs to $\Gamma_{Z_l}$, making the inclusion strict.
\end{proof}

So far we have shown that facets may be represented by non-facets in the canonical model and that truth sets of facets may be subsumed by larger facet truth sets. One might, therefore, ask a question whether it always happens. Could it be that all DMC sets are embedded into some maximal consistent sets, in which case the canonical model would be a pure simplicial model? We finish this section by dispelling this suspicion and showing that canonical simplicial model $\canmod$  is  impure.

\begin{exa}
Given at least three agents $a$, $b$, and $c$, consider simplicial model $\C$ in Figure~\ref{fig:bdead}
\begin{figure}[t]%
\centering
\begin{tikzpicture}[scale=0.5]
\node  (0) at (0, 0) {$v_a$};
\node  (1) at (4, 0) {$v_c$};

\draw [-] (0) to (1) node at (2,0.5) {$U$};
\end{tikzpicture}
	\caption{Simplicial model: agent $b$ is dead}
	\label{fig:bdead}
\end{figure}
\noindent
consisting of one facet $U = \{v_a,v_c\}$ of dimension 1. This simplicial model is impure in our terms because one of the agents, $b$, is dead in its only facet. Using $\Gamma \ce \{\phi \mid \C, U \vDash \phi\}$ for the set of formulae that are true in $U$ and recalling that $\Gamma$ is a DMC set, we are guaranteed that it is represented by $X_{\Gamma} = \{K_a \Gamma, K_c \Gamma\}$ in  canonical simplicial model  $\canmod$. We will now show that this $X_{\Gamma}$ is a canonical facet. 

Suppose towards a contradiction that there exists a DMC set $\Delta$ such that $K_a \Gamma = K_a \Delta$, $K_c \Gamma = K_c \Delta$, but $K_b \Delta \ne \varnothing$, meaning that $X_{\Gamma} \subsetneq X_\Delta$ is not a facet. Take any propositional variable $p_b \in P_b$. Since $K_b \phi \impdef p_b$ for any formula $\phi$ by Lemma~\ref{lem:eq_definabilities}\ref{from_K_to_atom}, it follows that~$p_b \Sin \D$, \ie~$p_b \in \D$ or $\lnot p_b \in \D$. If $p_b \in \D$, then, since $K_a \theta, p_b \impdef K_a p_b$ for any $\theta$ 
by Lemma~\ref{lem:eq_definabilities}\ref{lem:if_phi_then_ka_phi}, we conclude that $K_a p_b \Sin \D$ because $K_a \D = K_a \G \ne \varnothing$. The same applies if $\lnot p_b \in \D$ since~\mbox{$p_b \eqdef \lnot p_b$}. On the one hand, $K_a p_b$ cannot be in $K_a \D = K_a \G$ because $K_a p_b$ is undefined in~$U$ and, hence, cannot be true there. 
On the other hand, $\lnot K_a p_b$ cannot be  in $\D$ because this would imply that $K_a \lnot K_a p_b \in K_a \D$ by axiom~\textbf{5} and, hence, $K_a \lnot K_a p_b \in  K_a \G$, which is impossible for the same reason.
This contradiction shows that such a DMC set $\Delta$ does not exist, and $X_\Gamma$ is a canonical facet where agent $b$ is dead.
\end{exa}

\section{Conclusions and Future Work}
\label{conclusion}

We presented the axiomatization called $\Sfive^{\isdef}$ and showed that it is sound and complete for the class of impure complexes. The completeness proof involved the construction of the canonical simplicial model, a novel kind of a canonical model for a three-valued epistemic semantics where the third value means `undefined.' 

We envisage various generalizations and refinements, both concerning axiomatizations but also concerning  epistemic semantics for distributed computing. 

While $\Sfive^{\isdef}$ applies to any number of agents, for small numbers of agents, \eg,~for one or two agents, simpler axiomatizations should be  possible. As follows from Lemmatas~\ref{ex:counterK}--\ref{ex:noMP}, the full system is required starting from at least three agents. 

Different epistemic semantics can be envisioned wherein, apart from local variables, there are global variables stating that agents are alive. Those logics would be more expressive. Such more expressive semantics would allow for an embedding into two-valued semantics for impure complexes.

Adding distributed knowledge modalities is a generalization of the semantics, see, for instance,~\cite{lics23}. This would also increase the expressivity with respect to impure complexes. Distributed knowledge can be considered independently from global atoms for life and death. Finally, one can generalize from simplicial complexes to simplicial sets.

Such generalizations of the epistemic semantics {would} come with different axiomatizations, and also promise to address specific modeling issues in distributed computing.

\section*{Acknowledgment} The authors would like to thank Marta B\'\i{}lkov\'a, Giorgio Cignarale, Stephan Felber, Krisztina Fruzsa, Hugo Rinc\'on Galeana, and Thomas Schl\"ogl for many inspiring discussions and especially Ulrich Schmid for his  enthusiasm and knowledge, which appear to be equally inexhaustible. We also thank the anonymous reviewers for their comments and remarks.

\bibliographystyle{alphaurl}
\bibliography{lmcs_impure}

\appendix

\section{Auxiliary Statements}
\label{a}
\begin{lem}
\label{lem:many_imps}
The following rule is derivable:
	\[
	\displaystyle \orule{\vdash \phi_1 \imp \psi_1 \qquad\dots\qquad \vdash \phi_m \imp \psi_m}{ \vdash \phi_1 \land \dots \land \phi_m \imp \psi_1 \land \dots \land \psi_m}.\]
\end{lem}
\begin{proof}
The tautology
\[
(\phi_1 \imp \psi_1) \imp \Bigl(\cdots \imp (\phi_m \imp \psi_m) \imp (\phi_1 \land \dots \land \phi_m \imp \psi_1 \land \dots \land \psi_m)\Bigr)
\] 
can be used, followed 
by $m$ applications of \textbf{MP}$^{\bowtie}$ justified by the fact that the definability of~$\phi_1 \land \dots \land \phi_m \imp \psi_1 \land \dots \land \psi_m$ implies the definability of all $\phi_i$'s and all $\psi_j$'s, hence, of any other formula in the derivation.
\end{proof}

\begin{lem}
\label{lem:norm_K_many}
\[
\vdash\quad K_a \psi_1 \land \cdots \land K_a \psi_m\quad \to\quad K_a(\psi_1 \land \cdots \land \psi_m).
\]
\end{lem}
\begin{proof}
Induction on $m$. For $m=1$, the statement is trivial because $K_a \psi_1 \to K_a \psi_1$ is a tautology. For $m=2$, it follows from Lemma~\ref{lem:ConjOfKaImpKaOfConj}. Assume the statement holds for $m=k$.
\begin{enumerate}[1.]
\item
$K_a \psi_1 \land \cdots \land K_a \psi_k \to K_a(\psi_1 \land \cdots\land \psi_k)$ \hfill induction hypothesis
\item
$K_a \psi_{k+1} \to K_a \psi_{k+1}$ \hfill tautology
\item
$K_a \psi_1 \land \cdots \land K_a \psi_k \land K_a \psi_{k+1} \to K_a(\psi_1 \land \cdots \land \psi_k) \land K_a \psi_{k+1}$ \hfill Lemma~\ref{lem:many_imps}(1.,2.)
\item
$K_a(\psi_1 \land \cdots \land \psi_k) \land K_a \psi_{k+1}\to K_a(\psi_1 \land \cdots \land \psi_k \land \psi_{k+1})$ \hfill Lemma~\ref{lem:ConjOfKaImpKaOfConj}
\item
$K_a \psi_1 \land \cdots \land K_a \psi_{k+1} \to K_a(\psi_1 \land \cdots \land  \psi_{k+1})$ \hfill Lemma~\ref{lem:condHS}(3.,4.) \\
\strut\hfill as\quad$K_a \psi_1 \land \cdots \land K_a \psi_{k+1} \to K_a(\psi_1 \land \cdots \land \psi_{k+1})\quad \impdef\quad K_a(\psi_1 \land \cdots \land \psi_k) \land K_a \psi_{k+1}$ 
\end{enumerate}
\end{proof}

\section{Additional proofs}
\label{b}

\defcon*
\begin{proof}
\hfill
\begin{enumerate}[(a)]
\item
$\G \impdef \phi$ means that there is a finite subset $\G' \subseteq \G$ such that $\phi$ is defined whenever all formulae from $\G'$ are. But $\G'$ is  also a finite subset of $\G \cup \D$. Hence, $\G \cup \D \impdef \phi$.
\item
Since $\D \impdef \phi$, there exists a finite subset $\D' = \{\psi_1,\dots, \psi_m\} \subseteq \D$ such that $\phi$ is defined whenever all $\psi_i$'s are. Since $\G \impdef \D$, for each $\psi_i$ there must exist a finite subset $\G_i \subseteq \G$ such that $\psi_i$ is defined whenever all formulae from $\G_i$ are. It is now easy to see that $\phi$~is defined whenever all formulae from the finite subset $\bigcup_{i=1}^m \G_i \subseteq \G$ are. Hence, $\G \impdef \phi$.

\item Starting from this clause,  we consider an arbitrary simplicial model $\C = (C, \chi, \ell)$ and its face $X \in C$ and assume all formulae to the left of $\impdef$ to be defined in $X$, \eg,~here we assume that  $\C,X \isdef K_a \phi$.
	By Definition~\ref{def:defin}, there exists $Y \in C$ such that $\C,Y\isdef \phi$ and $a \in \chi(X \cap Y)$. Hence, $a \in \chi(X)$ and the conclusion follows from Definition~\ref{def:defin}.
\item
To show that $p_a \impdef K_a p_a$, assume that  $\C,X \isdef p_a$.
	By Definition~\ref{def:defin},  $a \in \chi(X)$. Equivalently, $a \in \chi(X \cap X)$. Hence, $\C, X \isdef K_a p_a$  by Definition~\ref{def:defin}. The other direction follows from clause~\ref{from_K_to_atom}.
\item 
Assume $\C, X \isdef K_a \phi$. By Definition~\ref{def:defin}, there exists $Y \in C$ such that $\C,Y\isdef \phi$ and $a \in \chi(X \cap Y)$. Then $a \in \chi(Y \cap Y)$ so that  $\C, Y \isdef K_a \phi$. Therefore, $\C, X \isdef K_a K_a \phi$.

Assume $\C, X \isdef K_a K_a \phi$. By Definition~\ref{def:defin}, there exists $Y \in C$ such that $\C,Y\isdef K_a \phi$ and $a \in \chi(X \cap Y)$. This in turn implies that there exists $Z \in C$ such that $\C,Z \isdef \phi$ and $a \in \chi(Y \cap Z)$. Then $a \in \chi(X \cap Z)$ and $\C, X \isdef K_a \phi$.
\item 
That $K_a \phi \eqdef \lnot K_a \phi$ follows from Definition~\ref{def:defin}. It remains to prove only that additionally $\lnot K_a \phi \eqdef K_a \lnot K_a \phi$.

Assume $\C, X \isdef \lnot K_a \phi$. Hence, $\C, X \isdef K_a \phi$.  
By Definition~\ref{def:defin}, there exists $Y \in C$ such that $\C,Y\isdef \phi$ and $a \in \chi(X \cap Y)$. Then $a \in \chi(Y \cap Y)$ so that $\C, Y \isdef   K_a \phi$ and $\C, Y \isdef   \lnot K_a \phi$.
Therefore, $\C, X \isdef K_a \lnot K_a \phi$.

Assume $\C, X \isdef K_a \lnot K_a \phi$. By Definition~\ref{def:defin}, there is $Y \in C$ such that $\C,Y\isdef \lnot K_a \phi$ and $a \in \chi(X \cap Y)$, hence, also $\C, Y \isdef K_a \phi$. This in turn implies that there exists~$Z \in C$ such that $\C,Z \isdef \phi$ and $a \in \chi(Y \cap Z)$. Then $a \in \chi(X \cap Z)$ and $\C, X \isdef K_a \phi$, hence, $\C, X \isdef \lnot K_a \phi$.
\item
Assume $\C, X \isdef K_a \theta$ and $\C, X \isdef \phi$. By Definition~\ref{def:defin}, there exists $Y \in C$ such that $\C,Y\isdef \theta$ and $a \in \chi(X \cap Y)$.  Then $a \in \chi(X \cap X)$ so that $\C, X \isdef  K_a \phi$.
\item
 Assume $ K_a \G, \phi \impdef K_a\psi$,
\ie~there is a finite subset $\G' \subseteq K_a \G \cup \{\phi\}$ such that $K_a \psi$~is defined  whenever all formulae from $\G'$ are. If $\phi \notin \G'$, then  $\G' \subseteq K_a \G \cup \{K_a \phi\}$, so we are done. Otherwise,  $\G' = K_a \Theta \cup \{\phi\}$ for some finite subset $\Theta \subseteq \G$.
Assume  $\C, X \isdef K_a \theta$ for all $K_a \theta \in K_a\Theta$ and $\C, X \isdef K_a \phi$.
The latter means that  there exists $Y \in C$ such that $a \in \chi(X \cap Y)$ and $\C,Y \isdef \phi$. It follows from the former that, for each $K_a \theta \in K_a\Theta$, there exists $Z_\theta \in C$ such that $a \in \chi(X \cap Z_\theta)$ and $\C, Z_\theta \isdef \theta$. Since face $X$ cannot have  two $a$-vertices,  we have $a \in \chi(Y \cap Z_\theta)$, making  $\C, Y \isdef K_a \theta$ for each~$K_a \theta \in K_a\Theta$. Recalling $\C,Y \isdef \phi$, we   conclude $\C,Y \isdef K_a \psi$. Therefore, there exists $Z \in C$ such that $a \in \chi(Y \cap Z)$ and $\C,Z \isdef \psi$. The uniqueness of  the $a$-vertex in $Y$ yields once again that  $a \in \chi(X \cap Z)$, hence, $\C,X \isdef K_a \psi$. We have shown that $K_a \psi$ is defined  whenever all formulae from the finite subset $K_a \Theta \cup \{K_a \phi\} \subseteq K_a\G \cup \{K_a \phi\}$ are. Hence, $K_a\G, K_a \phi \impdef K_a \psi$.
\item
Assume $ K_a \G, \lnot\phi \impdef K_a\psi$. Since $\phi \eqdef \lnot \phi$, also $ K_a \G, \phi \impdef K_a\psi$. By clause~\ref{def_raise_ka}, it follows that $K_a\G, K_a \phi \impdef K_a \psi$. Since $K_a \phi \eqdef \lnot K_a \phi$, we finally obtain $K_a\G, \lnot K_a \phi \impdef K_a \psi$.

\item 
Assume $\C, X \isdef K_a (K_a \psi_1 \land \cdots \land K_a \psi_m) \imp K_a \theta$. Then, $\C, X \isdef K_a (K_a \psi_1 \land \cdots \land K_a \psi_m)$ and $\C, X \isdef K_a \theta$. By Definition~\ref{def:defin},  there exists $Y\in C$ such that $\C, Y \isdef \theta$ and $a \in \chi (X \cap Y)$ and there exists $Z \in C$ such that $\C, Z \isdef K_a \psi_1 \land \cdots \land K_a \psi_m$ and $a \in \chi (X \cap Z)$. From the latter, we conclude that $\C, Z \isdef  K_a \psi_i$ for each $i$, meaning that for each $i$ there is~$Z_i \in C$ such that $\C, Z_i \isdef \psi_i$ and $a \in \chi(Z \cap Z_i)$. Since $a \in \chi(Y \cap Z_i)$ for each $i$, we have $\C, Y \isdef K_a \psi_i$ for each $i$. Consequently, $\C, Y \isdef K_a \psi_1 \land \cdots \land K_a \psi_m \imp  \theta$, and, as required, $\C, X \isdef K_a(K_a \psi_1 \land \cdots \land K_a \psi_m \imp  \theta)$.

\item 
The proof of this clause is the same as for the preceding clause because $K_a \theta \eqdef \M_a \theta$.

\item 
Assume $\C, X \isdef K_a \psi_1 \land \cdots \land K_a \psi_m$. Since $\C, X \isdef K_a \psi_1$, there is  $Z$ such that $\C, Z \isdef \psi_1$ and $a \in \chi(X \cap Z)$. Thus, $a \in \chi(X \cap X)$. Therefore, $\C, X \isdef K_a(K_a \psi_1 \land \cdots \land K_a \psi_m)$. 

Assume  $\C, X \isdef K_a(K_a \psi_1 \land \cdots \land K_a \psi_m)$. By Definition~\ref{def:defin}, there exists $Y \in C$ such that $\C, Y \isdef K_a \psi_1 \land \cdots \land K_a \psi_m$ and $a \in \chi(X \cap Y)$. It follows that $\C,Y \isdef K_a \psi_i$ for each~$i$, meaning that for each $i$ there exists $Z_i \in C$ such that $\C, Z_i \isdef \psi_i$ and $a \in \chi(Y \cap Z_i)$. Since we then have $a \in \chi(X \cap Z_i)$ for each $i$, we conclude that $\C, X \isdef K_a \psi_i$ for each~$i$. Thus $\C, X \isdef K_a \psi_1 \land \cdots \land K_a \psi_m$.
\qedhere
\end{enumerate} 
\end{proof}

\soundimp*
\begin{proof}
Additional cases  proved earlier in~\cite{vDitKuz22arXiv}.
\begin{itemize}
\item Substitution instances of propositional tautologies are valid because their truth tables produce truth whenever they produce anything.
\item Axiom {\bf L} is valid. 
Assume $\C,X \isdef K_a p_a \lor K_a \neg p_a$. Then $\C,X \isdef K_a p_a$, that is, there is a face~$Y$ with $a \in \chi(X \cap Y)$ such that $\C,Y \isdef p_a$. Hence, $a \in \chi(X)$ and $\C,X \isdef p_a$. Let $X_a \in X$ be the vertex of $X$ with color $a$, \ie~$\chi(X_a) = a$. 
 We now show that $\C,X \vDash K_a p_a \lor K_a \neg p_a$. Given $\C,X_a \isdef p_a$, either $\C,X_a \vDash p_a$ or  $\C,X_a\vDash \neg p_a$. Since $X_a$~is  the only $a$-colored vertex in any face $a$-adjacent to $X_a$,  in the former case,  $\C,X_a \vDash K_a p_a$ and in the latter case, $\C,X_a \vDash K_a \neg p_a$. Either way, $\C,X_a \vDash K_a p_a \lor K_a \neg p_a$. Since $\{X_a\} \subseteq X$, it now follows from Lemma~\ref{lem:defMonotonicity}\ref{monot:true} that $\C,X \vDash K_a p_a \lor K_a \neg p_a$, as required.

		\item Axiom {\bf T} is valid.  Assume $\C,X \isdef K_a\phi \to  \phi$. It follows that $\C, X \isdef \phi$ and $a \in \chi(X)$. If $\C,X \vDash K_a \phi$, then $\C,Y \isdef \phi$ implies $\C,Y \vDash \phi$ for all $a$-adjacent faces $Y$, in particular, for~$Y=X$ itself, which is $a$-adjacent because $a \in \chi(X\cap X)$. Since $\C,X \isdef \phi$, we conclude $\C,X \vDash \phi$. We have shown that $\vDash K_a \phi \to \phi$.
		\item  Axiom {\bf 4} is valid.  If  $\C,X \vDash K_a \phi$, then $\C,X\isdef K_a\phi$ by Lemma~\ref{lem:threevalues}\ref{clause:threevalues_one}. In particular, $\C, W \isdef \phi$ for some   $W$ such that $a \in \chi(X \cap W)$. We also conclude $\C,X \isdef K_a K_a\phi$ by Lemma~\ref{lem:eq_definabilities}\ref{def:trans} and, hence, $\C, X \isdef K_a \phi \to K_a K_a\phi$, so that we do not need to assume this separately.	
		In order to prove $\C,X \vDash K_a K_a \phi$, take an arbitrary  face $Y$ such that $a \in \chi(X\cap Y)$. Then $a \in \chi(Y \cap W)$ and, hence, $\C, Y \isdef K_a \phi$.
		To  show that $\C,Y \vDash K_a \phi$, take an arbitrary face  $Z$ such that $a \in \chi(Y\cap Z)$ and $\C,Z \isdef \phi$.  Then $a \in \chi(X\cap Z)$. Thus, $\C,Z \isdef \phi$ and the assumption $\C,X\vDash K_a \phi$ yield $\C,Z \vDash \phi$. We proved that $\phi$~is true whenever defined in faces $a$-adjacent to $Y$. Hence, $\C, Y \vDash K_a \phi$. This shows that  $K_a \phi$~is true whenever defined in faces $a$-adjacent to~$X$. In addition,  $\C, X \isdef K_a \phi$ and $a \in \chi(X \cap X)$. Hence, $\C, X \vDash K_a K_a \phi$. We have shown $\vDash K_a \phi \imp K_a K_a \phi$.
		\item    Axiom {\bf 5} is valid.  If  $\C,X \vDash \lnot K_a \phi$, then  $\C, W \vDash \lnot\phi$ for some   $W$ such that $a \in \chi(X \cap W)$ by Corollary~\ref{cor:simple_def}\ref{eq:falseK}. By Lemma~\ref{lem:threevalues}\ref{clause:threevalues_one}, both  $\C,X\isdef \lnot K_a\phi$ and $\C, W \isdef \lnot \phi$. The former implies $\C,X \isdef K_a \lnot K_a\phi$ by Lemma~\ref{lem:eq_definabilities}\ref{def:eucl}. The latter implies $\C, W \isdef  \phi$ by Definition~\ref{def:defin}.
Hence, $\C, X \isdef \lnot K_a \phi \to K_a \lnot K_a\phi$, so that we do not need to assume this separately.
		In order to prove $\C,X \vDash K_a \lnot K_a \phi$, take an arbitrary  face $Y$ such that $a \in \chi(X\cap Y)$. Since $a \in \chi(Y \cap W)$, it follows by Corollary~\ref{cor:simple_def}\ref{eq:falseK} that   $\C,Y \vDash \lnot K_a \phi$. Thus, we proved that $\lnot K_a \phi$~is true in all faces $a$-adjacent to $X$. Given that  $X$ is one of such faces, $\C,X \vDash K_a \lnot K_a \phi$. We have shown $\vDash \lnot K_a \phi \imp K_a \lnot K_a \phi$.

\item Necessitation rule {\bf N} preserves validity. Assume $\vDash \phi$  
 and  $\C,X \isdef K_a \phi$. Validity of $\phi$ means that it is true whenever defined. In particular, it is true whenever defined for all faces $a$-adjacent to $X$. In addition, $\C,X \isdef K_a \phi$ means that there is a face $a$-adjacent to~$X$ where $\phi$ is defined. Therefore, $\C,X \vDash K_a \phi$. We have shown that $K_a \phi$ is true whenever defined, as required.
	Thus, $ \vDash K_a \phi$.\qedhere
\end{itemize}
\end{proof}

\contrapos*
\begin{proof}
\hfill
\begin{enumerate}[1.]
\item $(\phi \to \psi) \to (\lnot \psi \to \lnot \phi)$ \hfill tautology
\item $\phi \to \psi$ \hfill derivable by assumption
\item $\lnot \psi \to \lnot \phi$ \hfill \textbf{MP}$^{\bowtie}$(1.,2.) as $\lnot \psi \to \lnot \phi \impdef \phi \to \psi$
\end{enumerate}
\end{proof}

\kamonot*
\begin{proof}
\hfill
	\begin{enumerate}[1.]
		\item $\phi \imp \psi$
		\hfill derivable by assumption
		\item $K_a(\phi \imp \psi)$
		\hfill \textbf{N}(1.)
		\item $K_a(\phi \imp \psi) \imp (K_a \phi \imp K_a \psi)$
		\hfill \textbf{K}$^{\bowtie}$ as $\psi , K_a \phi \impdef \phi$
		\item $K_a \phi \imp K_a \psi$
		\hfill \textbf{MP}$^{\bowtie}$(2.,3.) as $K_a \phi \imp K_a \psi \impdef K_a(\phi \imp \psi)$ follows from $\psi , K_a \phi \impdef \phi$
	\end{enumerate}
	Let us explain why $K_a \phi \to K_a \psi \impdef K_a(\phi \imp \psi)$ holds in the last step. If $K_a \phi \to K_a \psi$ is defined in a face $X$, then both $K_a \phi$ and  $K_a \psi$ are defined. The latter implies that there is an $a$-adjacent face $Y$ where $\psi$ is defined. As before, if $K_a \phi$ is defined in $X$, it is defined in all $a$-adjacent faces, including $Y$. Hence, the definability provision $\psi , K_a \phi \impdef \phi$ applied to~$Y$ yields that $\phi$, and, hence, $\phi \to \psi$, is also defined in $Y$, making $K_a (\phi \to \psi)$ defined in~$X$.
\end{proof}

\kaconj*
\begin{proof}
\hfill
	\begin{enumerate}[1.]
		\item $\phi \imp (\psi \imp \phi \land \psi)$
		\hfill tautology
		\item $K_a\bigl(\phi \imp (\psi \imp \phi \land \psi)\bigr)$
		\hfill \textbf{N}(1.) 
		\item $K_a\bigl(\phi \imp (\psi \imp \phi \land \psi)\bigr) \imp \bigl(K_a \phi \imp K_a(\psi \imp \phi \land \psi)\bigr)$
		\hfill axiom \textbf{K}$^{\bowtie}$ as $\psi \imp \phi \land \psi \impdef \phi$
		\item $K_a \phi \imp K_a(\psi \imp \phi \land \psi)$
		\hfill \textbf{MP}$^{\bowtie}$(2.,3.)\\\strut\hfill as\quad $K_a \phi \imp K_a(\psi \imp \phi \land \psi)\quad\impdef\quad K_a\bigl(\phi \imp (\psi \imp \phi \land \psi)\bigr)$
		\item $K_a(\psi \imp \phi \land \psi) \imp \bigl(K_a\psi \imp K_a(\phi \land \psi)\bigr)$
		\hfill axiom \textbf{K}$^{\bowtie}$ as $\phi \land \psi \impdef \psi$
		\item $K_a \phi \imp \bigl(K_a\psi \imp K_a(\phi \land \psi)\bigr)$ \hfill Lemma~\ref{lem:condHS}(4.,5.)\\\strut\hfill  as\quad  $K_a \phi \imp \bigl(K_a\psi \imp K_a(\phi \land \psi)\bigr)\quad \impdef\quad K_a(\psi \imp \phi \land \psi)$
		\item $\Bigl(K_a \phi \imp \bigl(K_a\psi \imp K_a(\phi \land \psi)\bigr)\Bigr) \imp \Bigl(K_a \phi \land K_a \psi  \imp K_a(\phi \land \psi)\Bigr)$ \hfill tautology
		\item $K_a \phi \land K_a \psi  \imp K_a(\phi \land \psi)$ \hfill \textbf{MP}$^{\bowtie}$(6.,7.)\\\strut\hfill as\quad $K_a \phi \land K_a \psi  \imp K_a(\phi \land \psi)\quad \impdef\quad K_a \phi \imp \bigl(K_a\psi \imp K_a(\phi \land \psi)\bigr)$
	\end{enumerate}
	
	Figure~\ref{fig:counterK} provides a counterexample to the opposite direction. Indeed, in simplicial model~$\C_{\mathbf{K}}$, formula $K_a(p_b \land p_c)$ is true, formula $K_a p_c$ is also true, but $K_a p_b$ is false because $p_b$~is false in~$X$  (where $p_c$ is undefined).
\end{proof}

\deductth*
\begin{proof}
\hfill
\begin{itemize}
\item 
Case  $\vdash \psi$. Then $\vdash \psi \to (\phi \to \psi)$ as a tautology, and $\vdash \phi \to \psi$ is obtained by~\textbf{MP}$^{\bowtie}$ because $\phi \to \psi \impdef \psi$. Hence, $\G \vdash \phi \to \psi$ by Definition~\ref{def:deriv}.
\item 
Case $\vdash \phi \to \psi$ for the finite subset $\{\phi\}\subseteq \G \cup \{\phi\}$. Then  $\G \vdash \phi \to \psi$ by Definition~\ref{def:deriv}.
\item
Case  $\vdash \bigwedge\G_0 \imp \psi$ for a finite non-empty subset $\G_0 \subseteq \G$.
\begin{enumerate}[1.]
\item 
$\bigwedge \G_0  \imp \psi$ \hfill assumption
\item
$\bigwedge\G_0 \land \phi \imp \bigwedge\G_0$ \hfill tautology
\item $\bigwedge\G_0 \land \phi \imp \psi$ \hfill 
Lemma~\ref{lem:condHS}(1.,2.) as $\bigwedge\G_0 \land \phi \imp \psi \impdef \bigwedge\G_0$
\item
$(\bigwedge\G_0 \land \phi \imp \psi) \imp \bigl(\bigwedge\G_0 \imp (\phi \imp \psi)\bigr)$ \hfill tautology
\item $\bigwedge\G_0 \imp (\phi \imp \psi)$ 
\hfill \textbf{MP}$^{\bowtie}$(3.,4.) as $\bigwedge\G_0 \imp (\phi \imp \psi) \impdef \bigwedge\G_0 \land \phi \imp \psi$
\end{enumerate}
Hence, $\G \vdash \phi \imp \psi$ by Definition~\ref{def:deriv} based on $\G_0\subseteq \G$.
\item
Case $\vdash \bigwedge\G_0 \land \phi \imp \psi$ for a finite subset $ \{\phi\} \sqcup \G_0 \subseteq \G$ such that $\G_0 \ne \varnothing$.
\begin{enumerate}[1.]
	\item $\bigwedge\G_0 \land \phi \imp \psi$ \hfill assumption 
	\item
	$(\bigwedge\G_0 \land \phi \imp \psi) \imp \bigl(\bigwedge\G_0 \imp (\phi \imp \psi)\bigr)$ \hfill tautology
	\item $\bigwedge\G_0 \imp (\phi \imp \psi)$ 
	\hfill \textbf{MP}$^{\bowtie}$(1.,2.) as $\bigwedge\G_0 \imp (\phi \imp \psi) \impdef \bigwedge\G_0 \land \phi \imp \psi$
\end{enumerate}
Hence, $\G \vdash \phi \imp \psi$ by Definition~\ref{def:deriv} based on $\G_0\subseteq \G$. \qedhere
\end{itemize}
\end{proof}

\addextraK*
\begin{proof}
\hfill
\begin{enumerate}[1.]
		\item $K_a \psi_i \imp K_a K_a \psi_i$ \qquad for each $i=1,\dots,m$\hfill axiom \textbf{4}
		\item $K_a \psi_1 \land \cdots \land K_a \psi_m \imp K_a K_a \psi_1 \land \cdots \land K_a K_a \psi_m$ \hfill Lemma~\ref{lem:many_imps}(1.)
		\item $K_a K_a \psi_1 \land \cdots \land K_a K_a \psi_m \imp K_a(K_a \psi_1 \land \cdots \land K_a \psi_m)$ \hfill Lemma~\ref{lem:norm_K_many}
		\item $K_a \psi_1 \land \cdots \land K_a \psi_m \imp K_a(K_a \psi_1 \land \cdots \land K_a \psi_m)$ \strut\hfill Lemma~\ref{lem:condHS}(2.,3.)\\\strut\hfill  as\quad $K_a \psi_1 \land \cdots \land K_a \psi_m \imp K_a(K_a \psi_1 \land \cdots \land K_a \psi_m)\quad \impdef\quad K_a K_a \psi_1 \land \cdots \land K_a K_a \psi_m$ 
		\end{enumerate}
\end{proof}

\derforcons*
\begin{proof}
\hfill
	\begin{enumerate}[1.]
		\item $\lnot (\phi \land K_a \psi_1 \land \cdots \land K_a \psi_m)$
		\hfill derivable by assumption
		\item $K_a \psi_1 \land \cdots \land K_a \psi_m \imp \lnot \phi$
		\hfill    equidefinable and equivalent to 1.
		\item $K_a\bigl(K_a \psi_1 \land \cdots \land K_a \psi_m \imp \lnot \phi\bigr)$
		\hfill \textbf{N}(2.)
		\item $K_a \bigl(K_a \psi_1 \land \cdots \land K_a \psi_m \imp \lnot \phi\bigr) \imp \bigl(K_a (K_a \psi_1 \land \cdots \land K_a \psi_m) \imp \M_a \lnot \phi \bigr)$
		\hfill axiom $\mathbf{K\M}$
		\item $K_a (K_a \psi_1 \land \cdots \land K_a \psi_m) \imp \M_a \lnot \phi$ \hfill 
		 \textbf{MP}$^{\bowtie}$(3.,4.) as, by Lemma~\ref{lem:eq_definabilities}\ref{lem:def_tough_hatKhat}, \\\strut\hfill$K_a (K_a \psi_1 \land \cdots \land K_a \psi_m) \imp \M_a \lnot \phi\quad \impdef\quad K_a \bigl(K_a \psi_1 \land \cdots \land K_a \psi_m \imp \lnot \phi\bigr)$
		\item $K_a \psi_1 \land \cdots \land K_a \psi_m \imp K_a(K_a \psi_1 \land \cdots \land K_a \psi_m)$\hfill Lemma~\ref{lem:add_extra_K}
		\item $K_a \psi_1 \land \cdots \land K_a \psi_m \imp \M_a \lnot \phi$ \hfill
		Lemma~\ref{lem:condHS}(5.,6.) as, by Lemma~\ref{lem:eq_definabilities}\ref{lem:def_tough_K},\\\strut\hfill  $K_a \psi_1 \land \cdots \land K_a \psi_m \imp \M_a \lnot \phi\quad\impdef\quad K_a (K_a \psi_1 \land \cdots \land K_a \psi_m)$ 
		\item $K_a \psi_1 \land \cdots \land K_a \psi_m \imp \lnot K_a \lnot \lnot \phi$ \hfill from 7. by definition of $\M_a$
		\item $\lnot \lnot K_a \lnot \lnot \phi \imp \lnot (K_a \psi_1 \land \cdots \land K_a \psi_m)$\hfill Lemma~\ref{lem:contrapos}(8.)
		\item $ K_a \lnot \lnot \phi \imp \lnot (K_a \psi_1 \land \cdots \land K_a \psi_m)$\hfill equidefinable and equivalent to 9.
		\item $\phi \to \lnot \lnot \phi$ \hfill tautology
		\item $K_a \phi \to K_a \lnot \lnot \phi$ \hfill Lemma~\ref{lem:syl5}(11.) as $\lnot \lnot \phi \impdef \phi$
		\item $K_a \phi \imp \lnot (K_a \psi_1 \land \cdots \land K_a \psi_m)$ \hfill Lemma~\ref{lem:condHS}(12.,10.)\\\strut\hfill   as\quad 
		$K_a \phi \imp \lnot (K_a \psi_1 \land \cdots \land K_a \psi_m)\quad\impdef\quad K_a \lnot \lnot \phi$
		\item $\lnot (K_a \phi \land K_a \psi_1 \land \cdots \land K_a \psi_m)$
		\hfill equidefinable and equivalent to  13. 
			\end{enumerate} 
\end{proof}

\derforconsneg*
\begin{proof}
	\hfill
	\begin{enumerate}[1.]
		\item $\lnot (\lnot\phi \land K_a \psi_1 \land \cdots \land K_a \psi_m)$
		\hfill derivable by assumption
		\item $K_a \psi_1 \land \cdots \land K_a \psi_m \imp \phi$
		\hfill equidefinable and equivalent to  1. 		
		\item $K_a(K_a \psi_1 \land \cdots \land K_a \psi_m \imp \phi)$
		\hfill \textbf{N}(2.)
		\item $K_a(K_a \psi_1 \land \cdots \land K_a \psi_m \imp \phi) \imp \bigl(K_a (K_a \psi_1 \land \cdots \land K_a \psi_m) \imp K_a \phi\bigr)$ \hfill \textbf{K}$^{\bowtie}$\\
		\strut \hfill as\quad  $K_a (K_a \psi_1 \land \cdots \land K_a \psi_m)\quad \impdef\quad K_a \psi_1 \land \cdots \land K_a \psi_m$ by Lemma~\ref{lem:eq_definabilities}\ref{lem:def_tough_K}
		\item $K_a (K_a \psi_1 \land \cdots \land K_a \psi_m) \imp K_a \phi$  \hfill \textbf{MP}$^{\bowtie}$(3.,4.) as, by Lemma~\ref{lem:eq_definabilities}\ref{lem:def_tough_hatK}, \\
		\strut\hfill$K_a (K_a \psi_1 \land \cdots \land K_a \psi_m) \imp K_a \phi\quad \impdef\quad K_a(K_a \psi_1 \land \cdots \land K_a \psi_m \imp \phi)$ 
		\item $K_a \psi_1 \land \cdots \land K_a \psi_m \imp K_a(K_a \psi_1 \land \cdots \land K_a \psi_m)$\hfill Lemma~\ref{lem:add_extra_K}
		\item $K_a \psi_1 \land \cdots \land K_a \psi_m \imp K_a \phi$ \hfill Lemma~\ref{lem:condHS}(5.,6.) as, by Lemma~\ref{lem:eq_definabilities}\ref{lem:def_tough_K}, \\\strut\hfill $K_a \psi_1 \land \cdots \land K_a \psi_m \imp K_a \phi \quad\impdef\quad K_a (K_a \psi_1 \land \cdots \land K_a \psi_m)$ 
		\item $\lnot (\lnot K_a \phi \land K_a \psi_1 \land \cdots \land K_a \psi_m)$
		\hfill equidefinable and equivalent to 7.  
	\end{enumerate} 
\end{proof}

\supsetincons*
\begin{proof}\hfill
\begin{enumerate}[(a)]
\item
\begin{enumerate}[1.]
	\item $\lnot\bigwedge \G \imp \lnot\left(\phi \land \bigwedge \G\right)$ \hfill tautology
	\item $\lnot\bigwedge \G$ \hfill derivable by assumption
	\item $\lnot\left(\phi \land \bigwedge \G\right)$ \hfill MP$^{\isdef}$(1.,2.) as $\lnot\left(\phi \land \bigwedge \G\right) \impdef \lnot\bigwedge \G$
\end{enumerate}
\item
\begin{enumerate}[1.]
	\item $\phi \land \bigwedge \G  \imp  \bigwedge \G $ \hfill tautology
	\item $\bigwedge \G \to \psi$ \hfill derivable by assumption
	\item $\phi \land \bigwedge \G \to \psi $ \hfill Lemma~\ref{lem:condHS}(1.,2.) as $\phi \land \bigwedge \G \to \psi \impdef \bigwedge \G$
\end{enumerate}
\end{enumerate} 
\end{proof}

\cancor*
\begin{proof}
Since any face of any simplicial model yields a DMC set by Lemma~\ref{lem:canmod_from_Faces}, $\LMC \ne \varnothing$. DMC~sets are not empty by definition. 
Since $\lnot \phi \impdef \phi$ and $\phi_1 \land \phi_2 \impdef \phi_i$ for $i=1,2$, it follows that, for each  $\G \in \LMC$, either $p_a \Sin\G$ or $K_a \psi \Sin \G$ for some formula $\psi$ and $a \in A$  because $\G$~is DM. 
Since $p_a \impdef K_a p_a$ by Lemma~\ref{lem:eq_definabilities}\ref{lem:eq_atom_K} and, hence, also $\lnot p_a \impdef K_a p_a$, we conclude that $K_a \psi \Sin \G$ for some formula $\psi$ and $a \in A$ in either case, from which it follows that $K_a \G \ne \varnothing$ for this $a$. Hence,  $X_\G \ne \varnothing$ for all $\G \in \LMC$. As a consequence, $C^c \ne\varnothing$ and every DMC set~$\G$ is represented in the canonical model by  $X_\G \in C^c$.
Set $C^c$ is closed under taking non-empty subsets and all its faces are non-empty by construction. The faces  are finite because for~$X \subseteq X_\G$, we have $|X| \leq |X_\G| \leq n+1$. Chromatic map $\chi^c$ is defined correctly because all vertices $K_a \G\in\VV^c$ are non-empty by construction and all formulae in $K_a \G$ start with $K_a$.  Thus, it cannot happen that $K_a \G = K_b \D \ne \varnothing$ for $a \ne b$, and the definition of~$\chi^c$ is unambiguous. 
Finally, by construction, each $X_\G$ (and, consequently, any $X \subseteq X_\G$) has at most one $K_a \G$ per agent $a$. Thus, the chromatic map condition is fulfilled. The correctness of~$\ell^c$ relies on the same argument as that of $\chi^c$, but in addition, by construction, $\ell^c(K_a \G) \subseteq P_a$  and   $\ell^c(K_a \G) = \ell^c(K_a \D)$ in case $K_a \G = K_a \D \ne \varnothing$ for $\G \ne \D$.
\end{proof}

\end{document}